\documentclass[12pt]{article}
\usepackage{graphicx}
\usepackage{natbib} %comment out if you do not have the package
\usepackage{url} % not crucial - just used below for the URL 
\usepackage{comment}
\newcommand{\blind}{0}

\usepackage{comment}

\usepackage{caption}
\usepackage{subcaption}
\usepackage{graphicx}

\usepackage{times,scalefnt}
\usepackage{amsmath,amssymb,mathtools}
\usepackage{tikz}
\usetikzlibrary{matrix,decorations.pathreplacing, calc, positioning, shapes.callouts, shapes.geometric}

\usepackage{amsmath}
\DeclareMathOperator*{\argmax}{arg\,max}
\DeclareMathOperator*{\argmin}{arg\,min}

\usepackage[boxruled,vlined,linesnumbered,commentsnumbered]{algorithm2e}

\usepackage{subcaption}

\newcommand{\p}{{p}}

\newcommand{\Cb}{{\pmb{C}}}

\newcommand{\y}{{\mathbf{y}}}
\newcommand{\model}{{\pmb{\eta}}}
\newcommand{\epsilonv}{{\pmb{\varepsilon}}}
\newcommand{\thetav}{{\pmb{\theta}}}
\newcommand{\Sigmav}{{\pmb{\Sigma}}}

\newcommand{\Ib}{{\mathbf{I}}}

\newcommand{\Kv}{{\mathbf{K}}}
\newcommand{\kv}{{\mathbf{k}}}
\newcommand{\Wv}{{\mathbf{W}}}
\newcommand{\wb}{{\mathbf{w}}}
\newcommand{\Sv}{{\mathbf{S}}}
\newcommand{\Bv}{{\mathbf{B}}}
\newcommand{\bv}{{\mathbf{b}}}
\newcommand{\xv}{{\mathbf{x}}}
\newcommand{\tauv}{{\pmb{\tau}}}
\newcommand{\phiv}{{\pmb{\phi}}}
\newcommand{\muv}{{\pmb{\mu}}}
\newcommand{\rhov}{{\pmb{\zeta}}}

\newcommand{\tcb}{\textcolor{black}}

\usepackage[english]{babel}
\usepackage{amsthm}
\newtheorem{theorem}{Theorem}[section]

\newtheorem{lemma}[theorem]{Lemma}

\date{}

\begin{document}

\def\spacingset#1{\renewcommand{\baselinestretch}%
{#1}\small\normalsize} \spacingset{1}

%%%%%%%%%%%%%%%%%%%%%%%%%%%%%%%%%%%%%%%%%%%%%%%%%%%%%%%%%%%%%%%%%%%%%%%%%%%%%%

\if0\blind
{
  \title{\bf Sequential Bayesian experimental design for calibration of expensive simulation models}
  \author{\"{O}zge S\"{u}rer\thanks{
    The authors gratefully acknowledge support from the NSF grant OAC 2004601}\hspace{.2cm}\\
    \small Department of Information Systems and Analytics, Miami University\\
    and \\
    Matthew Plumlee \thanks{The authors gratefully acknowledge support from the NSF grant  DMS 1953111} \\
    \small Department of Industrial Engineering and Management Sciences, Northwestern University \\
    and \\
    Stefan M. Wild \thanks{This material was based upon work supported by the U.S.\ Department of
	Energy, Office of Science, Office of Advanced Scientific Computing
	Research, applied mathematics and SciDAC programs under Contract No.\ \rm{DE-AC02-05CH11231}}. \\
    \small Applied Mathematics and Computational Research Division, Lawrence Berkeley National Laboratory}
  \maketitle
} \fi

\if1\blind
{
  \bigskip
  \bigskip
  \bigskip
  \begin{center}
    {\LARGE\bf Sequential Bayesian experimental design for calibration of expensive simulation models}
\end{center}
  \medskip
} \fi

\bigskip
\begin{abstract}
%The text of your abstract.  100 or fewer words.
   Simulation models of critical systems often have parameters that need to be calibrated using observed data.
   For expensive simulation models, calibration is done using an emulator of the simulation model built on simulation output at different parameter settings. 
   Using intelligent and adaptive selection of parameters to build the emulator can drastically improve the efficiency of the calibration process.
   The article proposes a sequential framework with a novel criterion for parameter selection that targets learning the posterior density of the parameters. 
   The emergent behavior from this criterion is that exploration happens by selecting parameters in uncertain posterior regions while simultaneously exploitation happens by selecting parameters in regions of high posterior density.
   The advantages of the proposed method are illustrated using several simulation experiments and a nuclear physics reaction model.
\end{abstract}

\noindent%
{\it Keywords:} Acquisition, Active learning, Emulation, Gaussian process, Uncertainty quantification %3 to 6 keywords, that do not appear in the title
\vfill

\newpage
\spacingset{1.8} % DON'T change the spacing!

\section{Introduction}

Across many engineering and science disciplines, simulation models are used to understand the properties and behaviors of complex systems.
These simulation models often include uncertain or unknown input parameters that affect the mechanics of the simulation. Calibration, the topic of this article, leverages observed data from the system to infer parameters and create predictions of quantities of interest. 
Bayesian calibration is a particular brand of calibration that provides uncertainty quantification on parameters and predictions. 
Bayesian calibration involves constructing a posterior distribution using  knowledge encapsulated in a prior distribution, which is typically a known, closed-form function of the parameters, and a likelihood, which requires the evaluation of the simulation model. 
The posterior distribution then reflects the relative probability of a set of parameters given both prior knowledge and the adherence to the observed data. 
In a standard Bayesian calibration process, Markov chain Monte Carlo (MCMC) computational techniques \citep{Gelman2004} are used to generate samples from the posterior using the unnormalized posterior density which represents the posterior density up to some unknown scaling constant. 
MCMC requires many evaluations of the likelihood for convergence of the posterior distribution.

Bayesian calibration using direct MCMC becomes difficult when a simulation model is computationally expensive. 
For example, a single simulation evaluation can take hours or more in fields such as epidemiology \citep{Yang2020} and nuclear physics \citep{Surer2022}. 
A popular solution is to build a computationally cheaper emulator to be used in place of the simulation model \citep{Ohagan2001, Higdon2004}. 
An emulator is built via flexible statistical models such as Gaussian processes (GPs) \citep{Rasmussen2005,gramacy2020surrogates} or Bayesian additive trees \citep{chipman2010bart} using simulation outputs recorded at a design consisting of a set of parameters. 
After an emulator is constructed, it is used in place of the expensive simulation model throughout the parameter space. 

Space-filling designs such as Latin Hypercube Sampling (LHS) or Sobol sequences are widely used in the literature to generate designs for general-purpose emulator construction (see \cite{santner2018design} for a detailed survey). 
The downside of space-filling designs for calibration is that they sample the parameter space without regard to the agreement between the simulation output and the observed data. 
This means space-filling designs often query the simulation model in regions of the parameter space with ill-fitting simulation output. 
This problem is compounded in multidimensional parameter spaces where high posterior density regions can have very small volume relative to the support of the prior distribution. 
As an illustration, consider the unimodal density function illustrated in Figure~\ref{illustration1} with two parameters in $[-4, 4] \times [-4, 4]$; see Section~\ref{sec:results} for details. The left panel of Figure~\ref{illustration1} shows 50 samples \tcb{using} LHS where most of the design parameters are placed in the regions with near-zero posterior density. Many of these parameters provide little information for the behavior of the simulation model near the region of interest.
\begin{figure}[t]
    \centering
    \begin{subfigure}{0.45\textwidth}
        \includegraphics[width=1\textwidth]{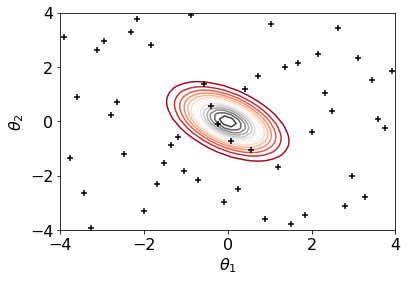}
    \end{subfigure}
    \begin{subfigure}{0.45\textwidth}
        \includegraphics[width=1\textwidth]{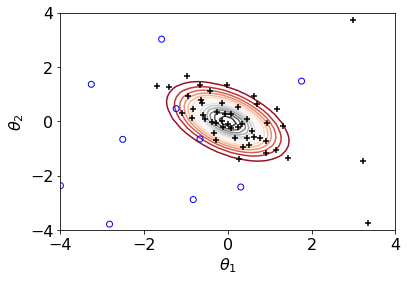}
    \end{subfigure}
    \caption{50 samples (plus markers) using LHS (left) and the sequential sampling procedure introduced in this article (right). Blue circles illustrate 10 samples used to initiate the proposed procedure. \tcb{The lines represent the level curves of the posterior.}}
    \label{illustration1}
\end{figure}

Sequential design (also called active learning) can be used to provide more precise inference by leveraging information learned to select the most informative parameters to sample. In sequential design, previous simulation evaluations are used to build a criterion that assesses the value of evaluating the simulation model at a new set of parameters. The criterion, also called an acquisition function, is then optimized to select a new set of parameters.
While previous research has shown that active learning paired with an emulator results in global high prediction accuracy \cite[Chapter 6]{gramacy2020surrogates}, there is less study on criteria in the calibration setting where global prediction is not needed.  
There is also a great deal of activity in Bayesian optimization \citep{Frazier2018}, where the goal is to find an optimizer (e.g., maximum a-posterior estimate), but this often results in designs that are hyper-localized and thus do not service calibration inference. 

This article proposes a novel sequential Bayesian experimental design procedure to accurately learn the unnormalized posterior density via an expected integrated variance (EIVAR) criterion. 
As a bit of foreshadowing, the right panel of Figure~\ref{illustration1} illustrates the results from our proposed approach, which generated 50 samples sequentially after 10 initial samples. 
Our approach places the majority of parameters in the high probability region, which facilitates learning of the posterior distribution. 
EIVAR measures the uncertainty in the estimate of the unnormalized posterior density, and the sequential strategy selects the parameters to minimize the overall uncertainty on the estimate. 
The resulting design allocates more points in the high probability regions than in the low probability regions, and hence the unnormalized posterior density function is  estimated more accurately with the proposed approach  than with space-filling approaches.
Moreover, we derive the proposed acquisition strategy considering both the high-dimensional simulation input and output since this is the case for most of the simulation models that require calibration.

Our EIVAR criterion is unique in how it leverages the natural emulation strategy used in modern calibration. The literature on sequential estimation of the posterior (or its analogs) often does not emulate the simulation model itself but instead emulate goodness-of-fit measures directly. 
For example, \cite{Joseph2015} and \cite{Joseph2019} propose an energy design criterion that aims to mimic a distribution while improving point diversity using an emulator over the distribution. Other examples of this perspective include \cite{Kandasamy2015}, who predict the (unnormalized) posterior, and \cite{Jarvenpa2019}, who instead predict the distance measure for approximate Bayesian computation \citep{Marin2012}. 
In the related topic of Bayesian quadrature methods \citep{Ohagan1999}, which aim to infer on the integral of a density function, there are methods that also use the emulation of the density itself (e.g., \cite{Fernandez2020}). 
Overall, these approaches seek to create good designs using the emulation of a goodness-of-fit measure itself such as the posterior. 
Building such emulators can be significantly more difficult compared to directly building emulators of the simulation model. 
There is untapped potential in directly using an emulator on the simulation model with its associated uncertainty quantification. 
One example of this perspective for the purposes of point estimation is \cite{Damblin2018}, who found success using an emulator for the simulation model to build an expected improvement criterion on the sum of the squared residuals.  % between the observed data and simulation output. 
In contrast to that work, our goal is to infer on the parameters' posterior density to provide both point estimation and uncertainty quantification for these parameters. 
\tcb{Recently, \cite{Koermer2023} propose an integrated mean-squared prediction error criterion for improved field prediction within the Kennedy and O’Hagan calibration framework. 
In this paper, our focus is on the parameters rather than on the field prediction.}% the potentials of such emulators are also utilized in this paper, we propose a novel, active parameter acquisition strategy based on the aggregated uncertainty on the unnormalized posterior density function to help exploring the entire posterior surface.  

The remainder of the paper is organized as follows. 
Section~\ref{sec:review} overviews the main steps of \tcb{the} sequential design procedure.
Section~\ref{sec:GPUQ} presents the resulting predictive mean and the variance of the unnormalized posterior density used to build our new acquisition strategy in Section~\ref{sec:acquisition}. 
We then extend the algorithm to the batched setting in Section~\ref{sec:batch} where several parameters are simultaneously evaluated in parallel to substantially reduce the total computation time. 
Finally, Section~\ref{sec:results} illustrates the computational and predictive advantages of the proposed approach using results from several simulation experiments including synthetic problems and \tcb{a} nuclear physics simulation model.

\section{Review and Overview}
\label{sec:review}

This section overviews the notation, Bayesian calibration methodology, and the proposed sequential algorithm.

\subsection{Simulation Structure and Notation}
\label{sec:notation}

Throughout this article, vectors and matrices are represented in boldface. \tcb{Our derivations consider a setting where simulation outputs are obtained at all design points once the simulation model is evaluated with parameters $\thetav$. To establish some notation, let $(\xv_1, \ldots, \xv_d)$ be the design points where the data $\y = (y(\xv_1), \ldots, y(\xv_d))$ is collected.
%The observed system data is represented with a vector $\y = (y_1, \ldots, y_d)$. 
The simulation model, represented with $\model$, is a function that takes the parameters $\thetav = (\theta_1, \ldots, \theta_p)$ in a space $\Theta \subset \mathbb{R}^p$ as inputs and returns a $d$-dimensional output $\model(\thetav) = (\eta(\xv_1, \thetav), \ldots, \eta(\xv_d, \thetav))$ \tcb{at design points $(\xv_1, \ldots, \xv_d)$}.}
We consider the statistical model
\begin{equation}
    \y = \model(\thetav) + \epsilonv,   \label{eq:statmodel}
\end{equation}
where $\epsilonv \sim {\rm MVN}(\mathbf{0}, \Sigmav)$ denotes the residual error and we presume $\Sigmav$ is a known matrix. \tcb{While the precise form of the residual error model is debated (for example, if a discrepancy term is included; see the articles of \cite{Ohagan2001}, \cite{Bayarri2007}, \cite{Plumlee2017}, and \cite{gu2018scaled} for discussions on this in a Bayesian context), for now we consider a known $\Sigmav$ accounts for the uncertainty in the difference between the data and the model. We extend our EIVAR criterion for the case of unknown $\Sigmav$ with a discrepancy term in Section~\ref{sec:KOH}.}

In the Bayesian calibration framework, the model parameters are viewed as random variables, and $\p(\thetav|\y)$ indicates the posterior probability density of the parameters $\thetav$ given the observed data $\y$. 
The prior knowledge about the unknown parameters $\thetav$ is represented by the prior probability density $p(\thetav)$. 
Based on Bayes' rule, the posterior density has the form 
\begin{equation} %\nonumber 
\label{eq:posterior}
    p(\thetav|\y) = \frac{ p(\y|\thetav) p(\thetav)}{\int_{\Theta}  p(\y|\thetav') p(\thetav') {\rm d \thetav'}} \propto \tilde{p}(\thetav|\y) =  p(\y|\thetav) p(\thetav),
\end{equation}
where $p(\y|\thetav)$ represents the likelihood and $\tilde{p}(\thetav|\y)$ represents the unnormalized posterior.  
Using the statistical model in (\ref{eq:statmodel}),
\begin{equation}
   \p(\y|\thetav) = (2 \pi)^{-d/2} |\Sigmav|^{-1/2} \exp\bigg(-\frac{1}{2} (\y - \model(\thetav))^\top \Sigmav^{-1} (\y - \model(\thetav )) \bigg), \label{eq:truelike}
\end{equation}
where it becomes clear that evaluation of $\p(\y|\thetav)$ requires evaluating the computationally expensive simulation output $\model(\thetav)$ at any desired $\thetav$. 

Our goal is to study estimation of the unnormalized posterior density, $\tilde{p}(\thetav|\y)$, as a function of $\thetav$. 
The unnormalized posterior represents the functional form of the posterior and is used to feed MCMC schemes to get posterior samples.
Throughout the remainder of this paper, $\tilde{p}(\thetav|\y)$ is referred to as the posterior for brevity.

\subsection{Emulation and Calibration}
\label{sec:roleGP}

Emulators are used to accelerate Bayesian inference by predicting the simulation output at novel parameter combinations. 
A particularly powerful class of emulators has been GPs \citep{santner2018design,gramacy2020surrogates}. As noted in the introduction, emulation of a simulation model is often easier compared to emulation of the posterior directly. Typically this is because transferring the simulation model's output $\model(\cdot)$ to obtain the likelihood in Equation~\eqref{eq:truelike} adds more complexity to the inference problem.

To understand the advantage of emulation of the simulation output, consider the one-dimensional example illustrated in Figure~\ref{illustration2}. Note that we remove the boldface characters from a parameter $\theta$, a model's output $\eta(\theta)$, and observed data $y$ in this example since inputs and outputs are both scalars. 
The left panel of Figure~\ref{illustration2} shows an emulator built to predict the simulation output. 
The emulator can directly predict the likelihood as well and get a variance of that prediction (more details are provided in Section~\ref{sec:GPUQ}) as in the middle panel. 
In contrast, an emulator designed directly to predict the likelihood $p(y|\theta)$ is shown in the right panel.
Figure~\ref{illustration2} shows that the prediction coming from the emulator of the simulation output is better at representing the complex posterior surface.
\begin{figure}[t]
    \centering
    \begin{subfigure}{0.32\textwidth}
        \includegraphics[width=1\textwidth]{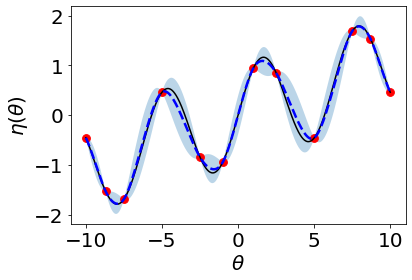}
    \end{subfigure}
    \begin{subfigure}{0.32\textwidth}
        \includegraphics[width=1\textwidth]{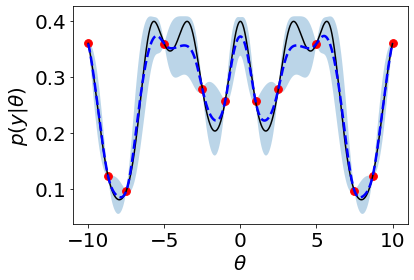}
    \end{subfigure}
    \begin{subfigure}{0.32\textwidth}
        \includegraphics[width=1\textwidth]{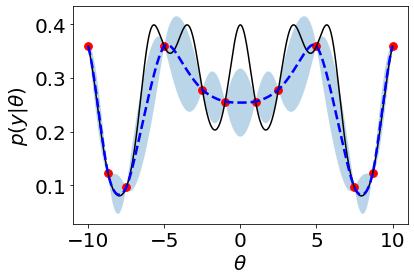}
    \end{subfigure}
    \caption{One-dimensional example of calibration with a single parameter $\theta$ and a simulation model $\eta(\theta) = \sin(\theta) + 0.1\theta$ and observation $y=0$ with observation variance one. Red dots indicate the simulation data used to build the emulators and the black line illustrates the simulation model or the likelihood. The blue dashed-line shows the prediction mean and the shaded area illustrates one predictive standard deviation from the mean. Shown are a GP emulator for $\eta(\cdot)$ (left panel); estimated $p(y|\cdot)$ using a GP emulator in left panel suggested in this article (middle panel); and a GP emulator directly for the likelihood $p(y|\cdot)$ (right panel).}
    \label{illustration2}
\end{figure}

\subsection{Overview of the Sequential Bayesian Experimental Design}
\label{sec:overview}

Our goal is to sequentially acquire simulation outputs from $n$ parameters to accurately estimate the posterior $\tilde{p}(\thetav|\y)$.
We first outline an approach where a single parameter and its simulation output are added to the existing simulation data at each stage indexed by $t$ as described in Algorithm~\ref{alg:oaat}.
This is extended to batched parameter collection in Section~\ref{sec:batch}.

\begin{algorithm}[H] \label{alg:oaat}

\spacingset{1} % for alg only
   \caption{Sequential Bayesian experimental design}
    
    \emph{Initialize} $\mathcal{D}_1 = \{(\thetav_{i}, \model(\thetav_{i})) : i = 1, \ldots, n_0\}$
    
    \For {$t = 1,\ldots,n$} {
        \emph{Fit} an emulator \tcb{$\model(\thetav)|\mathcal{D}_t$}%with $\mathcal{D}_t$
        
        \emph{Generate} candidate solutions $\mathcal{L}_t$
        
        \emph{Select} $\thetav^{\rm new} \in \argmin\limits_{\thetav^* \in \mathcal{L}_t} \mathcal{A}_t(\thetav^*)$

        \emph{Evaluate} $\model(\thetav^{\rm new})$

        \emph{Update} 
        $\mathcal{D}_{t+1} \gets \mathcal{D}_{t} \cup (\thetav^{\rm new}, \model(\thetav^{\rm new}))$

        }
\end{algorithm}

Sequential Bayesian experimental design has two main components: a Bayesian statistical model (in the form of a GP-based emulator on the simulation in this paper) and an acquisition function (which is often based on said model). 
Algorithm~\ref{alg:oaat} presents an overview of the process.
It begins by evaluating the simulation model at an initial experimental design of size $n_0$.
During each stage $t$, a new parameter is acquired for evaluation with the simulation model. 
The simulation data set $\mathcal{D}_{t+1} = \{(\thetav_{i}, \model(\thetav_{i})) : i = 1, \ldots, n_t\}$, where $n_t = n_0 + t$, includes all the simulation data obtained by the end of stage $t$. 
The simulation data, $\mathcal{D}_{t}$, is used to build the acquisition function $\mathcal{A}_t$.
At each stage $t$, the next parameter is chosen to minimize (approximately) the acquisition function $\mathcal{A}_t$.
The acquisition function is minimized over a candidate set of parameters $\mathcal{L}_t$ to avoid difficult numerical optimization of the proposed EIVAR acquisition function (see Section~\ref{sec:acquisition}). 
This new parameter and the simulation output is collected into $\mathcal{D}_{t+1}$ and the process repeats. Therefore, at each stage the most informative parameter is selected by leveraging information learned via simulation data $\mathcal{D}_{t}$.
   
\tcb{For the experiments presented in this paper, the sequential procedure is terminated once the simulation outputs are acquired from $n$ parameters. However, one can use alternative stopping criteria. One alternative is to decide on the best $n$ depending on the available computational resources (i.e., available computational hours, number of workers) especially if the simulation model is computationally expensive. If the computational budget allows or the simulation model is less expensive, one can create an hold-out data set (via some sampling techniques such as LHS), and then terminate the approach once a certain level of accuracy is achieved in predicting the posterior of the hold-out sample. Alternatively, one can obtain $k$-fold cross validation error using the simulation data $\mathcal{D}_t$, and the algorithm can be terminated if the decrease in the CV error stops improving.}

\subsection{Gaussian Process Model for the High-Dimensional Simulation Output}
\label{sec:GP}

A GP emulator is constructed at each stage of Algorithm~\ref{alg:oaat}. In this article, we show how this emulator can be used to construct an estimator for the posterior (see Section~\ref{sec:GPUQ}) and to build an acquisition function for the selection of a new parameter (see Section~\ref{sec:acquisition}). 
  
Since the emulation is over a $d$-dimensional output, where $d$ is the number of observations, we leverage a GP-based emulator that employs the basis vector approach that is now standard practice (e.g., \cite{Higdon2008}). \tcb{Similarly, for calibration, \cite{Huang2020} train separate independent emulators focused on each design point $\xv_i$ instead of building one big emulator for the entire $d$-dimensional output space. We note that one can replace the basis vector approach presented in this section with the approach presented by \cite{Huang2020}, and our results still hold with minor adjustments to the resulting statements and proofs. }
Let the $p \times n_t$ parameter matrix $\thetav_{1:n_t} = [\thetav_1, \ldots, \thetav_{n_t}]$ represent the parameters where the simulation model has been evaluated. 
The standardized outputs are stored in a $d \times n_t$ matrix $\Xi_t$ where the $i$th column is $(\model(\thetav_i) - \mathbf{h})/\mathbf{s}$ (computed elementwise), and $\mathbf{h}$ and $\mathbf{s}$ are the centering and scaling vectors for the simulation.
The $d \times q$ matrix $\Bv = [\bv_{1}, \ldots, \bv_{q}]$ stores the orthonormal basis vectors from principal component analysis of $\Xi_t$. Note that the dependence on $t$ is dropped from $\Bv$, $\mathbf{h}$, and $\mathbf{s}$ for simplicity.
Then, after transforming $\Xi_t$ with $\Wv_t =[\wb_{t,1}, \ldots, \wb_{t,q}] = \Bv^\top \Xi_t$ and letting $\mathbf{G} \coloneqq \text{diag}(\mathbf{s})$, an independent GP is used for each latent output $w_{j}(\cdot) = \bv_{j}^\top \mathbf{G}^{-1} \left(\model(\cdot) - \mathbf{h}\right)$ for $j = 1, \ldots, q$.
The GP model results in a prediction with mean $m_{t,j}(\cdot)$ and variance $\varsigma^2_{t,j}(\cdot)$ such that
\begin{equation}\label{eq:oned}
    w_{j}(\thetav)|\wb_{t,j} \sim \text{N}\left(m_{t,j}(\thetav), \varsigma^2_{t,j}(\thetav)\right), \qquad j = 1, \ldots, q.
\end{equation}
Following the properties of GPs, the mean $m_{t,j}(\thetav)$ and variance $\varsigma^2_{t,j}(\thetav)$ are given by
\begin{equation}
    \begin{aligned}
        & m_{t,j}(\thetav) = \kv_{j}(\thetav, \thetav_{1:n_t}) \Kv_{j}(\thetav_{1:n_t})^{-1} \wb_{t,j} \\
        & \varsigma^2_{t,j}(\thetav) = \kv_{j}(\thetav, \thetav) - \kv_{j}(\thetav, \thetav_{1:n_t}) \Kv_{j}(\thetav_{1:n_t})^{-1} \kv_{j}(\thetav_{1:n_t}, \thetav), \label{eq:meanvar_latent}
    \end{aligned}
\end{equation}
where
\begin{equation}
    \begin{aligned}
        k_j(\thetav, \thetav') &= \tau_{j}^{2} c(\thetav, \thetav'; \rhov_{j}),\\
        \Kv_{j}(\thetav_{1:n_t}) &= \kv_j(\thetav_{1:n_t}, \thetav_{1:n_t}') + \upsilon_{j}\Ib.
        \label{eq:covariance}
    \end{aligned}
\end{equation}
Here, $\tau_{j}^{2}$ represents a scaling parameter, $c(\thetav, \thetav'; \rhov_{j})$ is a correlation function, $\upsilon_{j} > 0$ is a nugget parameter, and $\rhov_{j} = (\zeta_{j, 1}, \ldots, \zeta_{j, p})$ is a lengthscale parameter. For the correlation function $c$, in this work we use the separable version of the Mat\'ern correlation function with smoothness parameter 1.5 \citep{Rasmussen2005} such that
\begin{equation}
    \begin{aligned}
        c(\thetav, \thetav'; \rhov_{j}) = & \prod_{l=1}^p\left[ (1 + |(\theta_l - \theta_l')\exp(\zeta_{j,l})|)  \exp\left(-\exp(\zeta_{j,l})|\theta_l - \theta_l'|\right)\right].
    \end{aligned}
\end{equation} 
%The details of mean $m_{t,j}(\cdot)$ and variance $\varsigma^2_{t,j}(\cdot)$ are given in Appendix~\ref{sec:app1}.

We define the $q$-dimensional vector $\mathbf{m}_t(\thetav) = (m_{t,1}(\thetav), \ldots, m_{t,q}(\thetav))$ and $q \times q$ diagonal matrix $\Cb_t(\thetav)$ with diagonal elements $\varsigma^2_{t,j}(\thetav)$ for $j = 1, \ldots, q$. Then, due to the pseudo-inverse of the orthonormal $\Bv$, the predictive distribution on the emulator output is
\begin{equation}
    \model(\thetav)|\mathcal{D}_{t} \sim \text{MVN}\left(\muv_t(\thetav), \Sv_t(\thetav)\right),
    \label{emu_final}
\end{equation}
where $\muv_t(\thetav) \coloneqq  \mathbf{h} + \mathbf{G} \Bv \mathbf{m}_t(\thetav)$ is the emulator predictive mean and $\Sv_t(\thetav) \coloneqq \mathbf{G} \Bv\Cb_t(\thetav)\Bv^\top \mathbf{G}$ is the covariance matrix. 

\section{Posterior Inference}
\label{sec:GPUQ}

This section describes the expected value and variance of $\tilde{p}(\thetav|\y)$ at each stage $t$ using the emulator described in Section~\ref{sec:GP}. 
Given a parameter $\thetav$, the emulator returns a predictive mean $\muv_t(\thetav)$ and covariance matrix $\Sv_t(\thetav)$ based on data $\mathcal{D}_{t}$. 
The following result allows us to compute the expectation $\mathbb{E}[\tilde{p}(\thetav|\y) | \mathcal{D}_{t}]$ and the variance $\mathbb{V}[\tilde{p}(\thetav|\y) | \mathcal{D}_{t}]$ of the posterior given $\mathcal{D}_{t}$. 
Throughout the paper, $f_\mathcal{N}(\mathbf{a}; \, \mathbf{b}, \, \mathbf{C})$ denotes the probability density function of the normal distribution with mean $\mathbf{b}$ and covariance $\mathbf{C}$, evaluated at the value $\mathbf{a}$.  
    \begin{lemma}\label{lemma:UQ}
    Assuming that the covariance matrices $\Sigmav$ and $\Sv_t(\thetav)$ are positive definite, under the model given by Equations~\eqref{eq:posterior}, ~\eqref{eq:truelike}, and ~\eqref{emu_final}, 
        \begin{gather}
            \mathbb{E}[\tilde{p}(\thetav|\y)| \mathcal{D}_{t}] = f_\mathcal{N}\left(\y; \, \muv_t(\thetav), \, \Sigmav + \Sv_t(\thetav)\right) p(\thetav), \label{expectedpostfinal}\\
            \mathbb{V}[\tilde{p}(\thetav|\y) |\mathcal{D}_{t}] = \left(\frac{1}{2^{d}\pi^{d/2}|\Sigmav|^{1/2}}f_\mathcal{N}\left(\y; \, \muv_t(\thetav), \, \frac{1}{2}\Sigmav + \Sv_t(\thetav)\right) \right. \nonumber \\
            \left. \hspace{2in} - \left(f_\mathcal{N}\left(\y; \, \muv_t(\thetav), \, \Sigmav + \Sv_t(\thetav)\right)\right)^2\right)p(\thetav)^2. \label{variancepostfinal}
        \end{gather}
    \end{lemma}
    \begin{proof}
        For the sake of brevity, we drop $\thetav$ from $\model(\thetav)$, $\muv_t(\thetav)$, and  $\Sv_t(\thetav)$.       
        By Equation~\eqref{eq:posterior}, since 
        \begin{equation*}
            \mathbb{E}[\tilde{p}(\thetav|\y)|\mathcal{D}_{t}] = \mathbb{E}[p(\y|\thetav) p(\thetav)| \mathcal{D}_{t} ] = \mathbb{E}[p(\y|\thetav)| \mathcal{D}_{t}]p(\thetav),
        \end{equation*}
        it suffices to show that
        \begin{align}\label{expectedlike}
            \mathbb{E}[p(\y|\thetav)|\mathcal{D}_{t}] = f_\mathcal{N}\left(\y; \, \muv_t, \, \Sigmav + \Sv_t\right).
        \end{align}
        Using Equations~\eqref{eq:truelike} and ~\eqref{emu_final}, 
        \begin{align}\notag%\label{jointlike}
            \begin{split}
                \mathbb{E}[p(\y|\thetav)|\mathcal{D}_{t}] &= \int f_\mathcal{N}\left( \y; \, \model, \, \Sigmav \right) f_\mathcal{N}\left(\model; \, \muv_t, \, \Sv_t\right) d \model, 
            \end{split}
        \end{align}
        which can be equivalently written as 
        \begin{align}\label{jointlike_ext}
            \begin{split}
                (2\pi)^{-d}|\Sigmav|^{-1/2} |\Sv_t|^{-1/2} \int \exp\left\{-\frac{1}{2}(\y - \model)^\mathsf{T}\Sigmav^{-1}(\y - \model) -\frac{1}{2}(\model - \muv_t)^\mathsf{T}\Sv_t^{-1}(\model - \muv_t)\right\} d \model.
            \end{split}
        \end{align}
        Defining $\mathbf{v} = \muv_t - \model$ and $\mathbf{z} = \y - \muv_t$, Equation~\eqref{jointlike_ext} becomes 
        \begin{align}\label{jointlikesubs} 
            (2\pi)^{-d} |\Sigmav \Sv_t|^{-1/2} \int \exp\left\{-\frac{1}{2}(\mathbf{v} + \mathbf{z})^\mathsf{T}\Sigmav^{-1}(\mathbf{v} + \mathbf{z}) -\frac{1}{2}\mathbf{v}^\mathsf{T}\Sv_t^{-1}\mathbf{v}\right\} d \mathbf{v}.
        \end{align} 
        Writing Equation~\eqref{jointlikesubs} in matrix notation,
        \begin{align}\notag%\label{jointlikematrix} 
            \begin{split}
                \mathbb{E}[p(\y|\thetav)|\mathcal{D}_{t}] &= (2\pi)^{-d} |\Sigmav \Sv_t|^{-1/2} \int \exp\left\{-\frac{1}{2}\left[{\begin{array}{c} \mathbf{v} \\
                \mathbf{z} \\\end{array}} \right]^\mathsf{T}\left[{\begin{array}{cc} \Sigmav^{-1} + \Sv_t^{-1} & \Sigmav^{-1} \\
                \Sigmav^{-1} & \Sigmav^{-1} \\
                \end{array} } \right] \left[{\begin{array}{c} \mathbf{v} \\
                \mathbf{z} \\ \end{array} } \right]\right\} d \mathbf{v} \\
                &= \int f_\mathcal{N}\left(\left[{\begin{array}{c} \mathbf{v} \\
                \mathbf{z} \\\end{array}} \right]; \,  \mathbf{0}, \, \left[ {\begin{array}{cc} \Sv_t & -\Sv_t \\
                                   -\Sv_t & \Sigmav + \Sv_t \\
                \end{array} } \right]  \right)d \mathbf{v},
            \end{split}
        \end{align}  
        which corresponds to marginalizing over $\mathbf{v}$. Equation~\eqref{expectedlike} then follows from Gaussian identities in the appendix of \cite[Equation~(A.6)]{Rasmussen2005}.
        
        We can use similar logic to compute the variance $\mathbb{V}[\tilde{p}(\thetav|\y)|\mathcal{D}_{t}]$ such that
        \begin{align}\label{variancepost} 
            \begin{split}
            \mathbb{V}[\tilde{p}(\thetav|\y)| \mathcal{D}_{t}] = \mathbb{V}[p(\y|\thetav)p(\thetav)| \mathcal{D}_{t}] &= \mathbb{V}[p(\y|\thetav)| \mathcal{D}_{t}]p(\thetav)^2 \\ &= \left(\mathbb{E}[p(\y|\thetav)^2| \mathcal{D}_{t}] - \mathbb{E}[p(\y|\thetav)| \mathcal{D}_{t}]^2\right)p(\thetav)^2.
            \end{split}
        \end{align}
        We first obtain
        \begin{align}\label{expectedsquared}
            \begin{split}
                & \mathbb{E}[p(\y|\thetav)^2| \mathcal{D}_{t}] = \int \left(f_\mathcal{N}\left(\y; \, \model, \, \Sigmav\right)\right)^2 f_\mathcal{N}\left(\model; \, \muv_t, \, \Sv_t\right) d \model \\
                &= \frac{1}{(2\pi)^{3d/2}|\Sigmav \Sv_t \Sigmav|^{1/2}} \int \exp\left\{-\frac{1}{2}\left(2(\y - \model)^\mathsf{T}\Sigmav^{-1}(\y - \model) + (\model - \muv_t)^\mathsf{T}\Sv_t^{-1}(\model - \muv_t)\right)\right\} d \model.
            \end{split}
        \end{align}
        Using again $\mathbf{v} = \muv_t - \model$ and $\mathbf{z} = \y - \muv_t$ and writing Equation~\eqref{expectedsquared} in matrix notation,
        \begin{align}\notag%\label{expectedsquaredmatrix} 
            \begin{split}
                 &= \frac{1}{(2\pi)^{3d/2} |\Sigmav|^{1/2} 2^{d/2}\left|\frac{1}{2} \Sigmav \Sv_t\right|^{1/2} }\int \exp\left\{-\frac{1}{2}\left[{\begin{array}{c} \mathbf{v} \\
                \mathbf{z} \\\end{array}} \right]^\mathsf{T}\left[{\begin{array}{cc} 2\Sigmav^{-1} + \Sv_t^{-1} & 2\Sigmav^{-1} \\
                2\Sigmav^{-1} & 2\Sigmav^{-1}\\
                \end{array} } \right] \left[{\begin{array}{c} \mathbf{v} \\
                \mathbf{z} \\ \end{array} } \right]\right\} d \mathbf{v} \\
                &= \frac{1}{2^{d}\pi^{d/2} |\Sigmav|^{1/2}}\int f_\mathcal{N}\left(\left[{\begin{array}{c} \mathbf{v} \\
                \mathbf{z} \\\end{array}} \right]; \,  \mathbf{0}, \, \left[ {\begin{array}{cc} \Sv_t & -\Sv_t \\
                                   -\Sv_t & \frac{1}{2}\Sigmav + \Sv_t \\
                \end{array} } \right]  \right) d \mathbf{v}.
            \end{split}
        \end{align} 
        The last equality holds since the determinant of the covariance matrix $\left[ {\begin{array}{cc} \Sv_t & -\Sv_t \\ -\Sv_t & \frac{1}{2}\Sigmav + \Sv_t \\ \end{array} } \right]$ is $\left|\frac{1}{2}\Sigmav \Sv_t \right|$.
        The integral corresponds to marginalizing over $\mathbf{v}$, which proves that 
        \begin{align}\label{expectedsquaredderived}
            \mathbb{E}[p(\y|\thetav)^2| \mathcal{D}_{t}] = \frac{1}{2^{d}\pi^{d/2}|\Sigmav|^{1/2}}f_\mathcal{N}\left(\y; \, \muv_t, \, \frac{1}{2}\Sigmav + \Sv_t\right).
        \end{align} 
        The substitution of Equation~\eqref{expectedsquaredderived} into Equation~\eqref{variancepost} yields Equation~\eqref{variancepostfinal}.
    \end{proof}
The expected value in Equation~\eqref{expectedpostfinal} is the estimate of the posterior density and the variance in Equation~\eqref{variancepostfinal} measures the uncertainty on the posterior. 
At the extreme, note that as $\Sv_t(\thetav)$ gets small and $\muv_t(\thetav)$ gets close to $\model(\thetav)$, meaning the emulator becomes more accurate, the posterior prediction will approach the actual posterior and the variance will approach zero. 
Our goal is to have this happen with as few samples as possible using smart selection of parameters to evaluate.
The prediction and the variance can now be used to construct our EIVAR acquisition function described in the next section.

\section{Expected Integrated Variance for Calibration}
\label{sec:acquisition}
   
The acquisition function $\mathcal{A}_t(\thetav)$ is used to measure the value of evaluating the simulation model $\model(\thetav)$ at a set of parameters $\thetav$ given data $\mathcal{D}_t$. 
We suggest using an acquisition function of the aggregated variance of the posterior over the parameter space. 
The proposed acquisition is built using the mean and variance of posterior prediction from Lemma~\ref{lemma:UQ}.
The EIVAR criterion is inspired by other criteria that aggregate variance under different statistical models  \citep{Seo200,Binois2019,Jarvenpa2019,Sauer2022}, but this functional form is novel because we leverage a statistical emulation of the simulation output with a calibration objective.
Specifically, the EIVAR criterion is calculated by
\begin{align} \label{acquisitionfunc}
    \begin{split}
        \mathcal{A}_t(\thetav^*) &= \int_{\Theta} \mathbb{E}_{\model^* | \mathcal{D}_{t}} \left( \mathbb{V}[p(\y|\thetav) \left| (\thetav^*, \model^*) \cup \mathcal{D}_{t} \right] \right) p(\thetav)^2 d \thetav ,
    \end{split}
\end{align}
where $\model^* \coloneqq \model(\thetav^*)$ represents the new simulation output at $\thetav^*$.
Notice that the expectation is taken over the hypothetical simulation output $\model^*$, which is a random variable because $\model^*$ is unknown under data $\mathcal{D}_t$. 

To unpack Equation~\eqref{acquisitionfunc}, consider Figure~\ref{illustration3}, which illustrates the three consecutive stages of sequential design using the example presented in Figure~\ref{illustration2}. 
For the left panel, using $n_0 = 12$ points and the emulator from Figure~\ref{illustration2}, the EIVAR criterion is calculated.
Since the minima of EIVAR is approximately $-6.4$, the simulation model is evaluated at this point.
After the statistical emulator of the simulation is refitted with this evaluation included, the EIVAR criterion is computed again at the next stage. 
The EIVAR values around $-6.4$ are no longer optimal (those points have the largest EIVAR values) compared to other candidate points, which discourages the reselection of these points.
In general, exploration under this criterion happens as previously evaluated parameters and surrounding areas have high EIVAR values.

\begin{figure}[t]
    \centering
    \begin{subfigure}{0.32\textwidth}
        \includegraphics[width=1\textwidth]{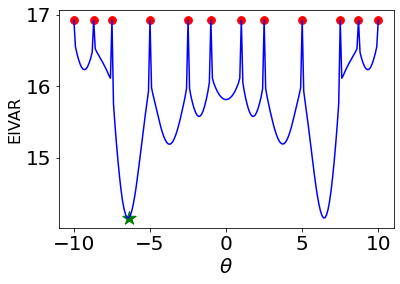}
        \label{stage1}
    \end{subfigure}
    \begin{subfigure}{0.32\textwidth}
        \includegraphics[width=1\textwidth]{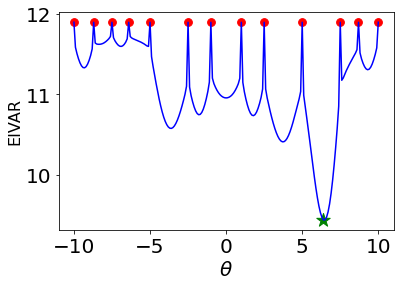}
        \label{stage2}
    \end{subfigure}
    \begin{subfigure}{0.32\textwidth}
        \includegraphics[width=1\textwidth]{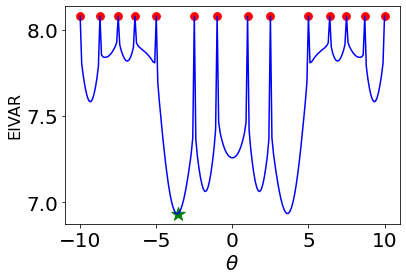}
        \label{stage3}
    \end{subfigure}
    \begin{subfigure}{0.32\textwidth}
        \includegraphics[width=1\textwidth]{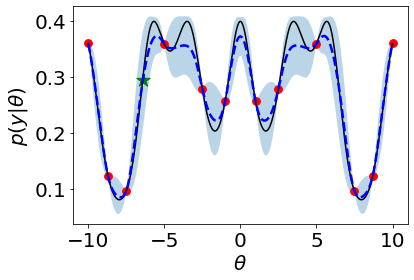}
        \label{stage1-like}
    \end{subfigure}
    \begin{subfigure}{0.32\textwidth}
        \includegraphics[width=1\textwidth]{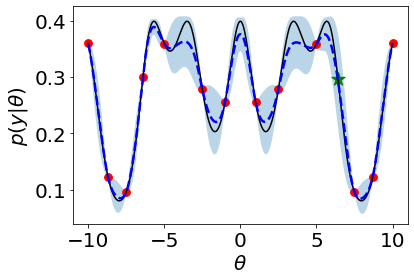}
        \label{stage2-like}
    \end{subfigure}
    \begin{subfigure}{0.32\textwidth}
        \includegraphics[width=1\textwidth]{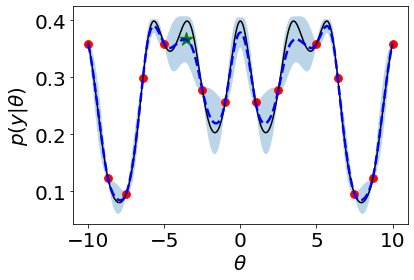}
        \label{stage3-like}
    \end{subfigure}
    \caption{One-dimensional example introduced in Figure~\ref{illustration2} to illustrate the EIVAR criterion in Equation~\eqref{acquisitionfunc} (first row). The blue line shows the EIVAR values across the range $\theta \in [-10, 10]$. Red dots illustrate the simulation data used to fit a GP emulator at each stage. The green star is the acquired point, which is added to the simulation data in the subsequent stage. Estimated and true $p(y|\theta)$ are illustrated in the second row.}
    \label{illustration3}
\end{figure}

To use Equation~\eqref{acquisitionfunc}, one needs to take the results from Lemma~\ref{lemma:UQ} and then integrate over all possible values of $\model^*$ given data $\mathcal{D}_t$.
To build closed-form estimates for this, we first establish general results for the GPs needed for the EIVAR criterion. 
Recall that $m_{t,j}(\thetav)$ and $\varsigma^2_{t,j}(\thetav)$ are the emulator mean and variance at stage $t$ for latent output $j = 1, \ldots, q$. 
After observing $\left(\thetav_{n_t+1},\model(\thetav_{n_t+1}) \right) = (\thetav^*,\model^*)$,
    \begin{align}\label{mean_future}
        \begin{split}
            m_{t+1,j}(\thetav) &= \left[\kv_j(\thetav, \thetav_{1:n_t}), \, k_j(\thetav, \thetav^*)\right] \left[ {\begin{array}{cc} \Kv_j(\thetav_{1:n_t}) & \kv_j(\thetav_{1:n_t}, \thetav^*) \\
            \kv_j(\thetav^*, \thetav_{1:n_t}) & k_j(\thetav^*, \thetav^*) + \upsilon_{j}\\
            \end{array} } \right]^{-1} \left[ {\begin{array}{cc} \wb_{t,j} \\
            w^*_j \\
                \end{array} } \right] \\
            &= m_{t,j}(\thetav) + \frac{\text{cov}_{t,j}(\thetav, \thetav^*)}{\varsigma^2_{t,j}(\thetav^*) + \upsilon_{j}} (w^*_j - m_{t,j}(\thetav^*)),
        \end{split} 
    \end{align}
where $\text{cov}_{t,j}(\thetav, \thetav^*) \coloneqq k_j(\thetav, \thetav^*) - \kv_j(\thetav, \thetav_{1:n_t}) \Kv_j^{-1}(\thetav_{1:n_t}) \kv_j(\thetav_{1:n_t}, \thetav^*)$, $w^*_j = \bv_{j}^\top \mathbf{G}^{-1} \left(\model^* - \mathbf{h}\right)$ is the value of the $j$th latent output for $\thetav^*$, \tcb{and $\upsilon_{j}$ is a nugget parameter introduced in Section~\ref{sec:GP}. }
Following similar logic for the variance function yields 
    \begin{align}\label{var_future}
        \begin{split}
            \varsigma^2_{t+1,j}(\thetav) &= \varsigma^2_{t,j}(\thetav) - \frac{\text{cov}_{t,j}(\thetav, \thetav^*)^2}{\varsigma^2_{t,j}(\thetav^*) + \upsilon_{j}}.
        \end{split} 
    \end{align}
The variance depends only on the chosen parameter $\thetav^*$, not the output $\model(\thetav^*)$. 
By taking the expected value and the variance of Equation~\eqref{mean_future},
    \begin{equation}
            \mathbb{E}_{\model^* | \mathcal{D}_{t}}\left[m_{t+1,j}(\thetav)\right] = m_{t,j}(\thetav) \qquad \mbox{and} \qquad
            \mathbb{V}_{\model^* | \mathcal{D}_{t}}\left[m_{t+1,j}(\thetav)\right] = \tau^2_{t,j}\left(\thetav, \thetav^*\right),
    \end{equation}
where $\tau^2_{t,j}(\thetav, \thetav^*) = \text{cov}_{t,j}(\thetav, \thetav^*)^2/(\varsigma^2_{t,j}(\thetav^*) + \upsilon_{j})$. Thus, since $w^*_j|\wb_{t,j} \sim \text{N}\left(m_{t,j}(\thetav^*), \, \varsigma^2_{t,j}(\thetav^*) + \upsilon_{j}\right)$ and by the transformation in Equation~\eqref{mean_future}, we conclude that 
    \begin{align}\label{mean_distr1}
        \begin{split}
            m_{t+1,j}(\thetav) | \mathcal{D}_t \sim \text{N}\left(m_{t,j}(\thetav), \, \tau^2_{t,j}(\thetav, \thetav^*)\right).
       \end{split} 
    \end{align}
    
For multidimensional simulation outputs, let $\tauv_t(\thetav, \thetav^*)$ be the $q \times q$ diagonal matrix with diagonal elements $\tau^2_{t,j}(\thetav, \thetav^*)$ and let $\phiv_{t}(\thetav, \thetav^*) \coloneqq \Bv\tauv_t(\thetav, \thetav^*)\Bv^\top$. 
Equation~\eqref{mean_distr1} implies 
    \begin{align}\label{mean_distr2}
        \begin{split}
            \muv_{t+1}(\thetav) | \mathcal{D}_t \sim \text{MVN}(\muv_{t}(\thetav), \phiv_{t}(\thetav, \thetav^*)).
       \end{split} 
    \end{align}  
In addition, Equation~\eqref{var_future} implies $\Cb_{t+1}(\thetav) = \Cb_t(\thetav) - \tauv_{t}(\thetav, \thetav^*)$ and $\Sv_{t+1}(\thetav) = \Sv_{t}(\thetav) - \phiv_{t}(\thetav, \thetav^*)$.
The next result is useful for computing the EIVAR criterion.

\begin{lemma}\label{prop:IEV}

Assuming that $ \phiv_{t}(\thetav, \thetav^*)$ is positive definite, under the conditions of Lemma~\ref{lemma:UQ}, 
    \begin{align} \label{eq:EIVcriterion}
        \begin{split}
            & \mathbb{E}_{\model^* | \mathcal{D}_{t}} \left( \mathbb{V}[p(\y|\thetav) \left| (\thetav^*, \model^*) \cup \mathcal{D}_{t} \right] \right) \\
            &= \frac{f_\mathcal{N}\left(\y; \, \muv_t(\thetav), \, \frac{1}{2}\Sigmav + \Sv_t(\thetav)\right)}{2^d \pi^{d/2} |\Sigmav|^{1/2}} - \frac{f_\mathcal{N}\left(\y; \, \muv_t(\thetav), \, \frac{1}{2}\left(\Sigmav + \Sv_t(\thetav) + \phiv_{t}(\thetav, \thetav^*)\right)\right)}{2^d \pi^{d/2} |\Sigmav + \Sv_t(\thetav) - \phiv_{t}(\thetav, \thetav^*)|^{1/2}}.
        \end{split}
    \end{align}
\end{lemma}
\begin{proof}
As in the previous results, we drop $\thetav$ from the notation and focus on
$$g^2_{t+1}(\thetav^*) =\mathbb{E}_{\model^* | \mathcal{D}_{t}} \left( \mathbb{V}[p(\y|\thetav) \left| (\thetav^*, \model^*) \cup \mathcal{D}_{t} \right] \right).$$
Let $\muv_{t+1}^*$ be the mean vector at $\thetav$ after seeing data $\mathcal{D}_{t+1}$ that includes $\left(\thetav_{n_t+1},\model(\thetav_{n_t+1}) \right) = (\thetav^*,\model^*)$. Lemma~\ref{lemma:UQ} shows $$\mathbb{V}[p(\y|\thetav) | \mathcal{D}_{t+1}] = \frac{1}{2^{d}\pi^{d/2}|\Sigmav|^{1/2}} f_\mathcal{N}\left(\y; \, \muv_{t+1}^*, \, \frac{1}{2}\Sigmav + \Sv_{t+1}\right) - \left(f_\mathcal{N}\left(\y; \, \muv_{t+1}^*, \, \Sigmav + \Sv_{t+1}\right)\right)^2.$$
Given that $\Sv_{t+1} = \Sv_t - \phiv_{t}(\thetav, \thetav^*)$ does not depend on $\model^*$, from Equation~\eqref{mean_distr2}
    \begin{align} \notag %\label{IEVpart3}
        \begin{split}
            g^2_{t+1}(\thetav, \thetav^*) &= \int \frac{1}{2^{d}\pi^{d/2}|\Sigmav|^{1/2}}f_\mathcal{N}\left(\y; \, \muv_{t+1}^*, \, \frac{1}{2}\Sigmav + \Sv_t - \phiv_t(\thetav, \thetav^*)\right) f_\mathcal{N}\left(\muv_{t+1}^*; \, \muv_{t}, \, \phiv_t(\thetav, \thetav^*)\right) d\muv_{t+1}^* \\& - \int \left(f_\mathcal{N}\left(\y; \, \muv_{t+1}^*, \, \Sigmav + \Sv_t - \phiv_t(\thetav, \thetav^*)\right)\right)^2 f_\mathcal{N}\left(\muv_{t+1}^*; \, \muv_{t}, \, \phiv_t(\thetav, \thetav^*)\right) d\muv_{t+1}^*.
        \end{split}
    \end{align}
Letting $\mathbf{L} = \frac{1}{2}\Sigmav + \Sv_t - \phiv_t(\thetav, \thetav^*)$, $\mathbf{M} =  \Sigmav + \Sv_t - \phiv_t(\thetav, \thetav^*)$, and $a_1 = \frac{2^{-d}\pi^{-d/2}|\Sigmav|^{-1/2}}{(2\pi)^d|\mathbf{L}\phiv_t(\thetav, \thetav^*)|^{1/2}}$, $a_2 = \frac{(2\pi)^{-3d/2}}{|\mathbf{M}\phiv_t(\thetav, \thetav^*)\mathbf{M}|^{1/2}}$, and assuming $\mathbf{L}$ and $\mathbf{M}$ are invertible, $g^2_{t+1}(\thetav, \thetav^*)$ is equivalently written as 
    \begin{align} 
        \begin{split} \label{eq:gnew}
        & a_1 \int \exp\left\{-\frac{1}{2} \left(\left(\y - \muv_{t+1}^*\right)^\top \mathbf{L}^{-1} \left(\y - \muv_{t+1}^*\right) + \left(\muv_{t+1}^* - \muv_{t}\right)^\top \phiv_t(\thetav, \thetav^*)^{-1} \left(\muv_{t+1}^* - \muv_{t}\right) \right)\right\}d\muv_{t+1}^* \\
        & - a_2 \int \exp\left\{-\frac{1}{2} \left(2\left(\y - \muv_{t+1}^*\right)^\top \mathbf{M}^{-1} \left(\y - \muv_{t+1}^*\right) + \left(\muv_{t+1}^* - \muv_{t}\right)^\top \phiv_t(\thetav, \thetav^*)^{-1} \left(\muv_{t+1}^* - \muv_{t}\right) \right)\right\}d\muv_{t+1}^*.
        \end{split}
    \end{align}
    Defining $\mathbf{v} = \muv_t - \muv_{t+1}^*$ and $\mathbf{z} = \y - \muv_t$, and writing Equation~\eqref{eq:gnew} in matrix notation
    \begin{align}\label{jointlikematrix} 
            \begin{split}
                =& \frac{2^{-d}\pi^{-d/2}|\Sigmav|^{-1/2}}{(2\pi)^d|\mathbf{L}\phiv_t(\thetav, \thetav^*)|^{1/2}} \int \exp\left\{-\frac{1}{2}\left[{\begin{array}{c} \mathbf{v} \\
                \mathbf{z} \\\end{array}} \right]^\mathsf{T}\left[{\begin{array}{cc} \mathbf{L}^{-1} + \phiv_t(\thetav, \thetav^*)^{-1} & \mathbf{L}^{-1} \\
                \mathbf{L}^{-1} & \mathbf{L}^{-1} \\
                \end{array} } \right] \left[{\begin{array}{c} \mathbf{v} \\
                \mathbf{z} \\ \end{array} } \right]\right\} d \mathbf{v} \\
                &- \frac{(2\pi)^{-3d/2}}{2^{d/2}|\mathbf{M}|^{1/2}\left|\frac{1}{2}\mathbf{M}\phiv_t(\thetav, \thetav^*)\right|^{1/2}} \int \exp\left\{-\frac{1}{2}\left[{\begin{array}{c} \mathbf{v} \\
                \mathbf{z} \\\end{array}} \right]^\mathsf{T}\left[{\begin{array}{cc} 2\mathbf{M}^{-1} + \phiv_t(\thetav, \thetav^*)^{-1} & 2\mathbf{M}^{-1} \\
                2\mathbf{M}^{-1} & 2\mathbf{M}^{-1} \\
                \end{array} } \right] \left[{\begin{array}{c} \mathbf{v} \\
                \mathbf{z} \\ \end{array} } \right]\right\} d \mathbf{v} \\
                =& \frac{1}{2^{d}\pi^{d/2}|\Sigmav|^{1/2}} \int f_\mathcal{N}\left(\left[{\begin{array}{c} \mathbf{v} \\
                \mathbf{z} \\\end{array}} \right]; \,  \mathbf{0}, \, \left[ {\begin{array}{cc} \phiv_t(\thetav, \thetav^*) & -\phiv_t(\thetav, \thetav^*)\\
                -\phiv_t(\thetav, \thetav^*) & \mathbf{L} + \phiv_t(\thetav, \thetav^*)\\
                \end{array} } \right]  \right)d \mathbf{v} \\
                &-\frac{1}{2^{d}\pi^{d/2}|\mathbf{M}|^{1/2}} \int f_\mathcal{N}\left(\left[{\begin{array}{c} \mathbf{v} \\
                \mathbf{z} \\\end{array}} \right]; \,  \mathbf{0}, \, \left[ {\begin{array}{cc} \phiv_t(\thetav, \thetav^*) & -\phiv_t(\thetav, \thetav^*)\\
                -\phiv_t(\thetav, \thetav^*) & \frac{1}{2}\mathbf{M} + \phiv_t(\thetav, \thetav^*)\\
                \end{array} } \right]  \right)d \mathbf{v}.
            \end{split}
        \end{align}
    The last equality holds since the determinants of the covariance matrices are $\left|\mathbf{L}\phiv_t(\thetav, \thetav^*)\right|$ and $\left|\frac{1}{2}\mathbf{M}\phiv_t(\thetav, \thetav^*)\right|$, respectively.
    Marginalizing over $\mathbf{v}$ and using the definition of $\mathbf{L}$ and $\mathbf{M}$ in Equation~\eqref{jointlikematrix}, we find that
    \begin{align} \notag %\label{IEVpart2}
        \begin{split}
            g^2_{t+1}(\thetav, \thetav^*) =& \frac{f_\mathcal{N}\left(\y; \, \muv_{t}, \, \frac{1}{2}\Sigmav + \Sv_t \right)}{2^{d}\pi^{d/2}|\Sigmav|^{1/2}} - \frac{f_\mathcal{N}\left(\y; \, \muv_{t}, \, \frac{1}{2}\left(\Sigmav + \Sv_t + \phiv_t(\thetav, \thetav^*)\right)\right)}{2^{d}\pi^{d/2}|\Sigmav + \Sv_t - \phiv_t(\thetav, \thetav^*)|^{1/2}} .
        \end{split}
    \end{align}
\end{proof}

At the $t$th stage, the suggestion is then to evaluate the simulation model at the parameter that minimizes Equation~\eqref{acquisitionfunc}.
In practice, the integral in Equation~\eqref{acquisitionfunc} is approximated with a sum over a uniformly distributed reference set $\Theta_{\rm ref}$ within the sample space, and thus the approximate EIVAR is given as
\begin{align} \label{approximateEIV}
    \begin{split}
        &  \frac{1}{|\Theta_{\rm ref}|} \sum_{\thetav \in \Theta_{\rm ref}} p(\thetav)^2 \left(\frac{f_\mathcal{N}\left(\y; \, \muv_{t}(\thetav), \, \frac{1}{2}\Sigmav + \Sv_t(\thetav)\right)}{2^{d}\pi^{d/2}|\Sigmav|^{1/2}} \right. \\
        & \hspace{1.5in} - \left. \frac{f_\mathcal{N}\left(\y; \, \muv_{t}(\thetav), \, \frac{1}{2}\left(\Sigmav + \Sv_t(\thetav) + \phiv_t(\thetav, \thetav^*)\right)\right)}{2^{d}\pi^{d/2}|\Sigmav + \Sv_t(\thetav) - \phiv_t(\thetav, \thetav^*)|^{1/2}}\right).
    \end{split}
\end{align}
The EIVAR criterion can be used to decide on the next point to be evaluated and it should be the minimizer from the candidate set $\mathcal{L}_t$. Since the first term in Equation~\eqref{approximateEIV} does not depend on $\thetav^*$, minimizing Equation~\eqref{approximateEIV} is equivalent to maximizing 
$$ \frac{1}{|\Theta_{\rm ref}|}\sum_{\thetav \in \Theta_{\rm ref}} p(\thetav)^2 \left( \frac{f_\mathcal{N}\left(\y; \, \muv_{t}(\thetav), \, \frac{1}{2}\left(\Sigmav + \Sv_t(\thetav) + \phiv_t(\thetav, \thetav^*)\right)\right)}{2^{d}\pi^{d/2}|\Sigmav + \Sv_t(\thetav) - \phiv_t(\thetav, \thetav^*)|^{1/2}}\right),$$
which is used to efficiently compute EIVAR in the experiments. \tcb{In this paper, the sequential design procedure with a closed-form EIVAR criterion is proposed to better estimate the parameters through estimating the posterior density. Additionally, since the posterior is well-estimated with our approach, one can use the emulator at the last stage of our sequential procedure (i.e., the emulator fitted with all acquired parameters) within the Bayesian calibration process to generate samples from the posterior of the parameters.}

In our examples and implementation, the matrices $\Sv_t(\thetav)$ and $\phiv_t(\thetav, \thetav^*)$ are rank at most $q$ due to projection with $\Bv$. 
If $q < d$, this is a violation of the positive definite conditions for Lemmas~\ref{lemma:UQ} and \ref{prop:IEV} which are adopted to streamline the proofs.
However, the final forms in the expressions for lemmas can be computed even without the positive definite condition on $\Sv_t(\thetav)$ and $\phiv_t(\thetav, \thetav^*)$  if $\Sigmav$ is positive definite.
In the case that $\Sigmav$ is not positive definite, we suggest adding a small positive definite matrix to  $\Sigmav$, which is similar to what is suggested by \cite{Higdon2008}.

\section{Additional Acquisition Functions}
\label{sec:additionalfuncs}

While Section \ref{sec:acquisition} describes our core proposed acquisition function, this section introduces alternative acquisition functions used in our empirical study in Section~\ref{sec:results}. 
Alternative acquisition functions are inspired from prior literature under different statistical models and represent reasonable alternatives.
One can consider them competitors from the literature, but they are morphed to fit our emulator model and calibration objective to obtain a fair benchmarking. 

The first alternative acquisition strategy is simply to choose the next point where the uncertainty of the posterior is highest. This is inspired from \cite{Seo200} who propose the selection of the parameter with the highest emulator variance for fitting a global emulator.
In the calibration context, \cite{Damblin2018} use the sum of squared errors of the residuals to measure the uncertainty on the simulation output prediction.
The analogous method called MAXVAR in the empirical study considers the uncertainty on the posterior prediction, and it replaces line~5 in Algorithm~\ref{alg:oaat} by
\begin{equation}\label{maxvar}
    \thetav^{\rm new} \in \argmax_{\thetav \in \mathcal{L}_t} \mathbb{V}[\tilde{p}(\thetav|\y)| \mathcal{D}_t].
\end{equation}
    
In addition to the posterior variance, one can utilize the expectation $\mathbb{E}[\tilde{p}(\thetav|\y)| \mathcal{D}_t]$ to build an acquisition function. 
In Bayesian quadrature scheme, \cite{Fernandez2020} propose using expectation with the diversity term $\min\limits_{\tilde{\thetav} \in \mathcal{D}_t} \|\thetav - \tilde{\thetav}\|_2$ to encourage exploration by fitting a GP for the posterior. In our proposed framework, the corresponding method abbreviated by MAXEXP uses
    \begin{equation}\label{maxexp}
        \thetav^{\rm new} \in \argmax_{\thetav \in \mathcal{L}_t} \mathbb{E}\left[\tilde{p}(\thetav|\y)| \mathcal{D}_t\right] \min\limits_{\tilde{\thetav} \in \mathcal{D}_t} \|\thetav - \tilde{\thetav}\|_2
    \end{equation}
in place of line~5 in Algorithm~\ref{alg:oaat}.

\tcb{In addition to the competitors introduced above, the two common acquisition functions are the expected improvement (EI) \citep{Jones1998} and the integrated mean squared error (IMSE) \citep{Sacks1989} in the active learning literature. EI is used to optimize a function in Bayesian optimization, whereas IMSE is used to build globally accurate emulators. In the calibration setting, one can consider maximizing the posterior $\tilde{p}(\thetav|\y)$ to find the maximum a-posterior estimate of the parameter via EI criterion. In such a case, EI tends to acquire parameters that are hyper-localized in the high posterior region and it provides minimal information about the entire posterior surface. In contrast, IMSE acquires many parameters at near-zero posterior region to better predict the entire simulation model and it does not bring any additional value for estimating the posterior. In Appendix~\ref{sec:ei_imse_eivar}, an illustration is provided to explain the details between EI, IMSE, and EIVAR.}

\section{EIVAR with Unknown Covariance}
\label{sec:KOH}

\tcb{The criterion introduced in Section~\ref{sec:acquisition} considers a known covariance matrix of errors between the data and the simulation model. This covariance matrix can incorporate observation or sensor noise, but many  modern analyses additionally account for the  uncertainty in the model itself in the form of a discrepancy.  A general representation is $\y = \model(\thetav) + \bv + \epsilonv,$
where $\bv = (b(\xv_1), \ldots, b(\xv_d))$ represents the discrepancy or the bias term at design points $(\xv_1, \ldots, \xv_d)$. One model, commonly attributed to \cite{Ohagan2001}, utilized a GP prior specification for both the simulation model $\model(\cdot)$ and the discrepancy function $b(\cdot)$.  
In our setting, this amounts to a multivariate normal distribution on $\bv$. The residual error $\epsilonv$ is often assumed sampled from a multivariate normal distribution with mean zero and covariance $\sigma^2_\varepsilon \mathbf{I}$. 
} 
%Here, $\bv$, $\sigma^2_\varepsilon$, and the best parameter $\thetav$ are all unknowns.
%Thus $\Sigmav^e$ captures the observation covariance $\Sigmav$ in addition to the covariance of the discrepancy term. 
\tcb{Let $\Sigmav^e$ denote the $d \times d$ covariance matrix of the discrepancy term $\bv$ plus the noise term $\epsilonv$.  Let $\thetav^e$ represent all ancillary parameters to define this matrix, which includes GP hyperparameters for the discrepancy term as well as $\sigma_\varepsilon^2$. 
In the case when there is no discrepancy between the simulation model and the observed data, one can consider that $\sigma^2_\varepsilon$ is $\thetav^e$ and set $\Sigmav^e = \sigma^2_\varepsilon \mathbf{I}$. }
\tcb{ 
As the ancillary parameters are unknown, and upon giving them the Bayesian treatment, we find joint posterior distribution is 
\begin{equation} \notag
    p(\thetav, \thetav^e|\y)  \propto  p(\y|\thetav, \thetav^e) p(\thetav)p(\thetav^e),
\end{equation}
where the likelihood, $p(\y|\thetav, \thetav^e)$, matches the right hand side of Equation~\eqref{eq:truelike} with $\Sigmav$ replaced by $\Sigmav^e$.} %in 
%\begin{equation} \notag
%   \p(\y|\thetav, \thetav^e) = (2 \pi)^{-d/2} |\Sigmav^e|^{-1/2} \exp\bigg(-\frac{1}{2} (\y - \model(\thetav))^\top \left(\Sigmav^e\right)^{-1} (\y - \model(\thetav )) \bigg). %\label{eq:truelike}
%\end{equation}}
%Let $\Sigmav^e = \sigma^2_\varepsilon \mathbf{I} + \mathbf{K}^b$ with $\mathbf{K}^b_{i,j} = \sigma^2_b \exp({-\lambda|x_i - x_j|)}$, and $\thetav^e = [\sigma^2_\varepsilon, \sigma^2_b, \lambda]$.

\tcb{Following modularization suggestions by \cite{Bayarri2007, Bayarri2009}, we propose to obtain plausible estimates of $\thetav^e$ at each stage $t$, and then to act as if these were fixed during acquisition. To do this, first, the simulation data is used to build an emulator for the simulation model as described in Section~\ref{sec:GP}. Then, the unknown ancillary parameters $\thetav^e$ are estimated by maximizing the likelihood, i.e.,
\begin{equation}
    \hat{\thetav}{}_{t}^{e}, \hat{\thetav}_t = \argmax_{\thetav^e \in \Theta^e, \thetav \in \Theta} |\Sigmav^e(\thetav^e) + \Sv_t(\thetav)|^{-1/2} \exp\bigg(-\frac{1}{2} (\y - \muv_t(\thetav))^\top \left(\Sigmav^e(\thetav^e) + \Sv_t(\thetav)\right)^{-1} (\y - \muv_t(\thetav)) \bigg),
    \label{eq:bias_like}
\end{equation}
%\begin{equation}
%    \hat{\thetav}{}_{t}^{e} = \argmax_{\thetav^e \in \Theta^e} \max_{ \thetav \in \Theta}   |\Sigmav^e(\thetav^e) + \Sv_t(\thetav)|^{-1/2} \exp\bigg(-\frac{1}{2} (\y - \muv_t(\thetav))^\top \left(\Sigmav^e(\thetav^e) + \Sv_t(\thetav)\right)^{-1} (\y - \muv_t(\thetav)) \bigg),
%    \label{eq:bias_like}
%\end{equation}
where $\Theta^e$ represents the space of ancillary  parameters. 
These estimates are used to approximate the posterior distribution of the calibration parameters, $\tilde{p}(\thetav|\y)$, by $\tilde{p}(\thetav|\y, \hat{\thetav}{}_{t}^{e})$. 
Then, the EIVAR in Equation~\eqref{approximateEIV} can be computed by replacing $\Sigmav$ with the covariance estimate $\Sigmav^e (\hat{\thetav}{}_{t}^{e})$. 
Although our focus is not on the prediction of new, unseen field data, one can use the estimates of the ancillary parameters $\thetav^e$ and alongside the posterior distribution of $\thetav$ to predict new, unseen field data in the case of the unknown discrepancy. } 

\section{Batch-Sequential Bayesian Experimental Design}
\label{sec:batch}

Section~\ref{sec:notation} to this point focuses on a sequential strategy where only one simulation is evaluated at a time. 
This section extends those ideas to another algorithm that can run multiple expensive simulations in parallel in a batched setting to reduce the wall-clock time needed to estimate the posterior. 
This batched setting presumes access to a fixed set of workers to perform simulation evaluations in parallel and an additional worker to acquire new parameters at each stage $t$. 

A batch synchronous procedure acquires a batch of size $b$ parameters at each stage, and waits until $b$ simulation evaluations are completed to acquire the next batch of $b$ parameters. 
Batch synchronous Bayesian optimization methods have been widely used (see, e.g., \cite{Ginsbourger2010} and \cite{Azimi2010}).
Figure~\ref{fig:batchsynch} illustrates an example batch synchronous procedure with the number of workers to perform simulation evaluations \tcb{equal} to the batch size. In this example, once the simulation evaluations of the jobs with indices $\{1,2,3, 4\}$ and $\{5,6,7, 8\}$ are received from the workers, the acquisitions $\mathcal{A}_1$ and $\mathcal{A}_2$ are initiated on worker with index $0$, respectively (see the same color jobs and acquisitions in the left panel of Figure~\ref{fig:batchsynch}). Each of the acquisitions $\mathcal{A}_1$ and $\mathcal{A}_2$ generate four parameters, and their associated jobs $\{5, 6, 7, 8\}$ and $\{9, 10, 11, 12\}$ and simulation evaluations are allocated to the workers (see the same color jobs and acquisitions in the right panel of Figure~\ref{fig:batchsynch}). 

\begin{figure}
    \begin{subfigure}{0.5\textwidth}
        \begin{tikzpicture}[scale=1/2, transform shape, y=1.5cm,
            worker/.style={circle, minimum width=0.5cm, minimum height=1cm, anchor=north west, text width=0.5cm, align=right, draw},
            batch/.style={minimum height=1cm, anchor=north west, align=right, draw},
            ]
            \tikzstyle{every node}=[font=\Large]
            \node[worker] at (0.0,4.0) {0};
            \node[worker] at (0.0,3.0) {1};
            \node[worker] at (0.0,2.0) {2};
            \node[worker] at (0.0,1.0) {3};
            \node[worker] at (0.0,0.0) {4};
            \node[batch, fill={lightgray}, minimum width=1.583cm, text width=1.583cm] at (1.0,3.0) {$4$};
            \node[batch, fill={lightgray}, minimum width=2.051cm, text width=2.051cm] at (1.0,1.0) {$2$};
            \node[batch, fill={lightgray}, minimum width=2.441cm, text width=2.441cm] at (1.0,0.0) {$1$};
            \node[batch, fill={lightgray}, minimum width=3.789cm, text width=3.789cm] at (1.0,2.0) {$3$};
            \node[batch, fill={pink}, minimum width=1.669cm, text width=1.669cm] at (7.0,0.0) {$5$};
            \node[batch, fill={pink}, minimum width=1.944cm, text width=1.944cm] at (7.0,3.0) {$8$};
            \node[batch, fill={pink}, minimum width=2.437cm, text width=2.437cm] at (7.0,2.0) {$7$};
            \node[batch, fill={pink}, minimum width=2.5cm, text width=2.5cm] at (7.0,1.0) {$6$};
            \node[batch, fill={white}, minimum width=1.0cm, text width=1.0cm] at (11.0,3.0) {$12$};
            \node[batch, fill={white}, minimum width=1.005cm, text width=1.005cm] at (11.0,1.0) {$10$};
            \node[batch, fill={white}, minimum width=2.097cm, text width=2.097cm] at (11.0,2.0) {$11$};
            \node[batch, fill={white}, minimum width=3.cm, text width=3.5cm] at (11.0,0.0) {$9$};
            \node[batch, fill={lightgray}, minimum width=2.0109999999999997cm, text width=2.0109999999999997cm] at (4.888999999999999,4.0) {$\mathcal{A}_1$};
            \node[batch, fill={pink}, minimum width=1.3cm, text width=1.3cm] at (9.6,4.0) {$\mathcal{A}_2$};   
        \end{tikzpicture}
        \label{batch1}
    \end{subfigure}
    \begin{subfigure}{0.5\textwidth}
        \begin{tikzpicture}[scale=1/2, transform shape, y=1.5cm,
            worker/.style={circle, minimum width=0.5cm, minimum height=1cm, anchor=north west, text width=0.5cm, align=right, draw},
            batch/.style={minimum height=1cm, anchor=north west, align=right, draw},
            ]
            \tikzstyle{every node}=[font=\Large]
            \node[worker] at (0.0,4.0) {0};
            \node[worker] at (0.0,3.0) {1};
            \node[worker] at (0.0,2.0) {2};
            \node[worker] at (0.0,1.0) {3};
            \node[worker] at (0.0,0.0) {4};
            \node[batch, fill={white}, minimum width=1.583cm, text width=1.583cm] at (1.0,3.0) {$4$};
            \node[batch, fill={white}, minimum width=2.051cm, text width=2.051cm] at (1.0,1.0) {$2$};
            \node[batch, fill={white}, minimum width=2.441cm, text width=2.441cm] at (1.0,0.0) {$1$};
            \node[batch, fill={white}, minimum width=3.789cm, text width=3.789cm] at (1.0,2.0) {$3$};
            \node[batch, fill={lightgray}, minimum width=1.669cm, text width=1.669cm] at (7.0,0.0) {$5$};
            \node[batch, fill={lightgray}, minimum width=1.944cm, text width=1.944cm] at (7.0,3.0) {$8$};
            \node[batch, fill={lightgray}, minimum width=2.437cm, text width=2.437cm] at (7.0,2.0) {$7$};
            \node[batch, fill={lightgray}, minimum width=2.5cm, text width=2.5cm] at (7.0,1.0) {$6$};
            \node[batch, fill={pink}, minimum width=1.0cm, text width=1.0cm] at (11.0,3.0) {$12$};
            \node[batch, fill={pink}, minimum width=1.005cm, text width=1.005cm] at (11.0,1.0) {$10$};
            \node[batch, fill={pink}, minimum width=2.097cm, text width=2.097cm] at (11.0,2.0) {$11$};
            \node[batch, fill={pink}, minimum width=3cm, text width=3.5cm] at (11.0,0.0) {$9$};
            \node[batch, fill={lightgray}, minimum width=2.0109999999999997cm, text width=2.0109999999999997cm] at (4.888999999999999,4.0) {$\mathcal{A}_1$};
            \node[batch, fill={pink}, minimum width=1.3cm, text width=1.3cm] at (9.6,4.0) {$\mathcal{A}_2$};    
        \end{tikzpicture}
        \label{batch2}
    \end{subfigure}
    \caption{Illustration of batch synchronous procedure with a batch of four parameters. Circles and rectangles represent workers and jobs, respectively, and the numbers inside denote the indices of the corresponding workers and jobs. The left panel illustrates the jobs initiating the same-color acquisition. Once the simulation evaluations of the jobs with indices $\{1, 2, 3, 4\}$ and $\{5, 6, 7, 8\}$ are received, the acquisitions $\mathcal{A}_1$ and $\mathcal{A}_2$ are initiated, respectively. The right panel shows the jobs resulting from the same-color acquisition. The acquisitions $\mathcal{A}_1$ and $\mathcal{A}_2$ generate the jobs with indices $\{5, 6, 7, 8\}$ and $\{9, 10, 11, 12\}$ that are allocated to the workers.}
    \label{fig:batchsynch}
\end{figure}
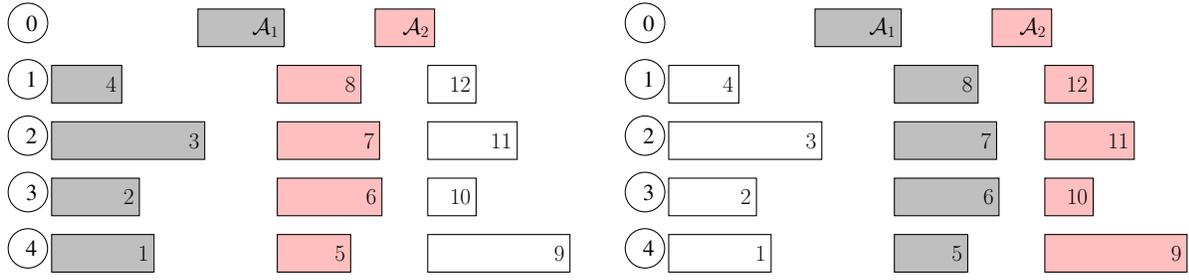

\begin{algorithm}[H] \label{alg:batch}
   
    \spacingset{1} % for alg only
    \caption{Batch-sequential Bayesian experimental design}
        \emph{Initialize} $\mathcal{D}_1 = \{(\thetav_{i}, \model(\thetav_{i})) : i = 1, \ldots, n_0\}$
        
        \For {$t = 1,\ldots,n/b$} {
        
            $\widehat{\mathcal{D}}_{t} \gets \mathcal{D}_{t}$; fit an emulator \tcb{$\model(\thetav)|\widehat{\mathcal{D}}_{t}$} %with $\widehat{\mathcal{D}}_{t}$
            
                \For {$n_a = 1,\ldots, b$} {
                    \emph{Generate} candidate solutions $\mathcal{L}_t$
                    
                    \emph{Select} $\thetav^{\rm new} \in \argmin\limits_{\thetav^* \in \mathcal{L}_t} \mathcal{A}_{t,\widehat{\mathcal{D}}_{t}}(\thetav^*)$
                    
                    \emph{Update} $\widehat{\mathcal{D}}_t \gets \widehat{\mathcal{D}}_t \cup (\thetav^{\rm new}, \muv_t(\thetav^{\rm new}))$
                    
                    \emph{Update} the emulator fitted in line~3 with $\widehat{\mathcal{D}}_t$
    
                    }
            $\mathcal{D}_{t+1} \gets \mathcal{D}_{t}$
            
            \emph{Update} $\mathcal{D}_{t+1}$ with $b$ new parameters and their simulation outputs

            }
\end{algorithm}

Algorithm~\ref{alg:batch} describes the batch synchronous procedure where one has access to $b$ workers to run $b$ simulations in parallel, whereas Algorithm~\ref{alg:oaat} allows to run only one simulation at a time. At each stage $t$, we receive $b$ completed jobs from workers (lines 10) and then acquire $b$ new parameters iteratively (lines 4--8). This process is repeated $n/b$ times (assuming that $n/b$ is an integer value) to acquire $n$ parameters. Once $b$ new parameters are acquired at each stage, the corresponding jobs to obtain their simulation outputs are allocated to the workers. In Algorithm~\ref{alg:batch}, a total of $n/b$ emulators are fitted from scratch. \tcb{Additionally, as the simulation output linked to the parameters is not yet obtained while constructing a batch, the emulator is updated using $\muv_t(\thetav^{\rm new})$ as emulation input (see lines 7-8), which is fitted at the start of each stage (see line 3).}
Imputing the current prediction as emulation input when building a batch is termed the kriging believer strategy in \cite{Ginsbourger2010} for Bayesian optimization. \tcb{We note that the kriging believer strategy is used to illustrate the propose batch-sequential procedure, and one can replace the kriging believer strategy with an alternative approach such as constant liar strategy.}
In Algorithm~\ref{alg:batch}, the hyperparameters for the emulator are not altered on line~8; only the covariance matrix is updated with data $\widehat{\mathcal{D}}_t$ and the prediction mean remains unchanged as the emulator is updated using $\muv_t(\thetav^{\rm new})$. The hyperparameters of the emulator are relearned at the beginning of each stage when fitting an emulator from scratch (line~3) after $b$ simulation outputs are received from workers during the previous stage. In Algorithm~\ref{alg:batch}, in addition to stage index $t$, $\widehat{\mathcal{D}}_t$ is a subscript of the acquisition function $\mathcal{A}_{t,\widehat{\mathcal{D}}_{t}}$ on line~6 to indicate that \tcb{each time} a new parameter is selected the acquisition function is minimized using the emulator based on data $\widehat{\mathcal{D}}_t$.

\section{Results}
\label{sec:results}

This section compares the performance of sequential and batch-sequential approaches with various acquisition functions described in Sections~\ref{sec:acquisition} and \ref{sec:additionalfuncs}. Section~\ref{sec:syntheticresults} considers the one-at-a-time sequential procedure on a variety of test functions. 
Section~\ref{sec:physicsexample} focuses on a computational physics problem and examines both sequential and batch-sequential approaches with batch sizes of $2, 4, 8, 16,$ and $32$. We provide an open source implementation of different selection strategies and their parallel implementations in the \texttt{PUQ} Python package \citep{PUQpackage}. All the test functions and the physics examples are provided in the \texttt{examples} directory within the \texttt{PUQ} Python package.

\subsection{Simulations from Synthetic Probability Densities}
\label{sec:syntheticresults}

\begin{figure}[t]
    \centering
    \begin{subfigure}{0.25\textwidth}
        \includegraphics[width=1\textwidth]{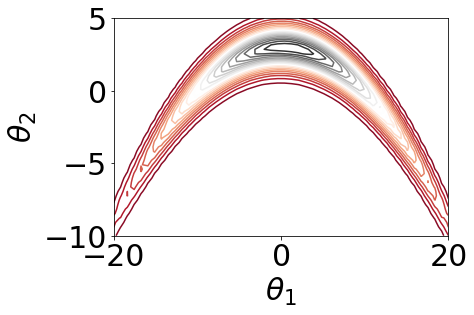}
    \end{subfigure}
    \begin{subfigure}{0.22\textwidth}
        \includegraphics[width=1\textwidth]{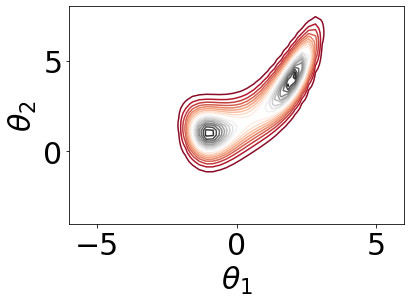}
    \end{subfigure}
    \begin{subfigure}{0.25\textwidth}
        \includegraphics[width=1\textwidth]{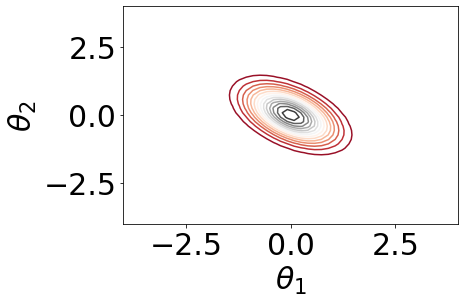}
    \end{subfigure}
    \begin{subfigure}{0.26\textwidth}
        \includegraphics[width=1\textwidth]{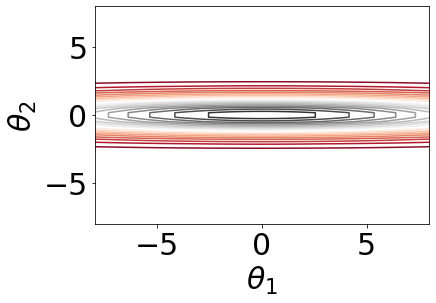}
     \end{subfigure}
    \caption{Posterior probability density illustration for four synthetic test functions. Labels of the test functions from left to right: banana, bimodal, unimodal, and unidentifiable.}
    \label{synth_figs}
\end{figure}
We investigate the performance of the sequential Bayesian experimental design with the proposed EIVAR acquisition function. As a baseline, we sample parameters from a uniform prior, and this approach is denoted as RND. 

To achieve a baseline outside of our framework, the results  are also compared with the minimum energy design (MINED) \citep{Joseph2015, Joseph2019}. Instead of minimizing the acquisition function over a candidate set, MINED uses a simulated annealing algorithm to minimize its acquisition function, and the total number of simulation evaluations to acquire $n$ parameters is larger than $n$, which makes it difficult to directly compare resulting designs; \tcb{thus,} the results are included in \tcb{Appendix~\ref{sec:app3}}. Results indicate that EIVAR performs better than MINED, and MINED has larger variability than EIVAR. We also note that MINED requires evaluating the simulation model \tcb{at more} than 1500 parameters to acquire 200 of them, whereas the proposed approach only requires $n + n_0$ evaluations (which is 210 in these experiments). Our article presumes that the simulation model is expensive to evaluate, meaning MINED is a potentially wasteful strategy relative to our proposed approach.

More reasonable competitors for a comparison use the sequential framework from Section~\ref{sec:overview} with the MAXVAR and MAXEXP acquisition functions described in Section~\ref{sec:additionalfuncs}. \tcb{When acquiring a set of parameters via MAXEXP criterion, since the diversity term depends on the distance between two sets of parameters, it is computed after the parameters are transformed into a unit hypercube.} 
We examine these different acquisition procedures using four common synthetic test functions (see examples in \cite{Joseph2019, Jarvenpa2019}) with different shapes of density illustrated in Figure~\ref{synth_figs}. 
In those examples, the parameter space is two-dimensional ($p = 2$) and the simulation model returns a two-dimensional output ($d = 2$) for the banana, bimodal, and unidentifiable functions, and a one-dimensional output ($d = 1$) for the unimodal function. 
We also examine performance with higher dimensional input parameters and simulation outputs, namely three-dimensional ($p = d = 3$), six-dimensional ($p = d = 6$), and ten-dimensional ($p = d = 10$) inputs and outputs.
The details of each simulation model and resulting likelihood are provided in Appendix~\ref{sec:app2}.

%\tcb{Moreover, in the case that the prior density depends on the scale of parameters and/or the numeric instability is a concern when building an emulator, we suggest users to transform their parameters into a unit hypercube $[0, 1]^p$.}
A uniform prior is assumed for all cases and the initial sample used to build the first emulator is taken from this prior with ranges given in Appendix~\ref{sec:app2}. \tcb{It is important to note that the prior probability density $p(\thetav)$ plays a critical role in EIVAR as in Equation~\eqref{approximateEIV} similar to the other Bayesian calibration procedures, and we provide an example in Appendix~\ref{sec:priorsensitivity} to test the sensitivity of EIVAR for various strengths of the prior.} We set the initial sample size $n_0 = 10$ for the two-dimensional functions in Figure~\ref{synth_figs} and set $n_0 = 20$, $n_0 = 40$, and $n_0 = 60$ for three-, six-, and ten-dimensional examples, respectively. The size of initial sample is chosen to allow the emulator to learn the response surface well enough during the initial stages and at the same time to allow improvements in the predictions through acquisitions.
At each stage $t$, the candidate list $\mathcal{L}_t$ is constructed as a random sample from the uniform prior with a size of $100$ for all the synthetic functions. 
For the banana, bimodal, unidentifiable, and unimodal functions, $200$ total parameters are acquired. For higher dimensional functions, $400$ total parameters are acquired.
To compare the performance of different acquisition functions, at each stage we report the mean absolute difference (MAD) between the estimated posterior and the ground truth posterior. 
In the experiments, we generate a set of reference parameters $\Theta_{\rm ref}$ and obtain the unnormalized posterior density for each parameter in this set. 
We then compute MAD at each stage $t$ as ${\rm MAD}_t(\tilde{p}, \hat{p}) = \frac{1}{|\Theta_{\rm ref}|} \sum_{\theta \in \Theta_{\rm ref}} |\tilde{p}(\thetav|\y) - \hat{p}_t(\thetav|\y)|$, where $\tilde{p}(\thetav|\y)$ is the unnormalized posterior density and $\hat{p}_t(\thetav|\y)$ is the estimated unnormalized posterior at stage $t$, which is obtained via Equation~\eqref{expectedpostfinal}. 
For the examples in Figure~\ref{synth_figs}, $\Theta_{\rm ref}$ is generated in a two-dimensional grid of $50^2$ points. For higher dimensional examples, since the size of a grid covering the parameter space becomes very large, $10^4$ points are generated \tcb{using} LHS. The set $\Theta_{\rm ref}$ is also used in Equation~\eqref{approximateEIV} to approximate the integral. 
\begin{figure}[t]
    \centering
    \begin{subfigure}{0.4\textwidth}
        \includegraphics[width=1\textwidth]{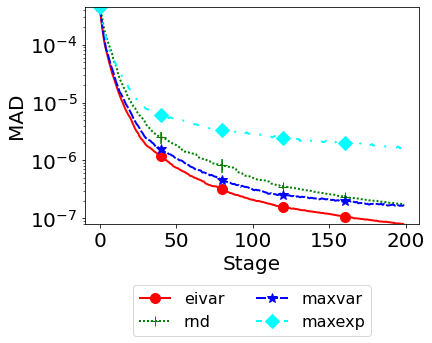}
    \end{subfigure}
    \begin{subfigure}{0.4\textwidth}
        \includegraphics[width=1\textwidth]{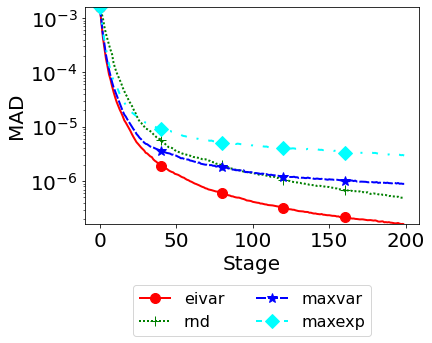}
    \end{subfigure}
    \begin{subfigure}{0.4\textwidth}
        \includegraphics[width=1\textwidth]{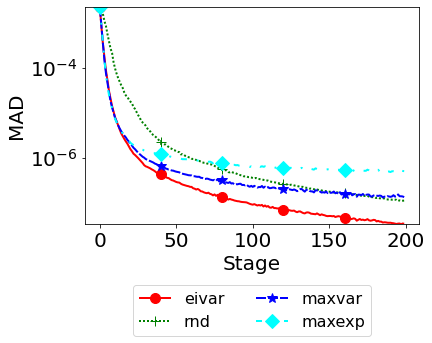}
    \end{subfigure}
    \begin{subfigure}{0.4\textwidth}
        \includegraphics[width=1\textwidth]{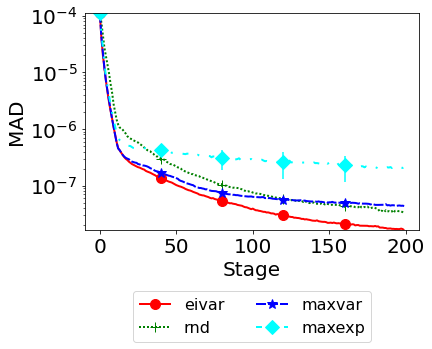}
     \end{subfigure}
    \caption{Comparison of different acquisition functions using the test functions illustrated in Figure~\ref{synth_figs} containing banana (upper-left), bimodal (upper-right), unimodal (lower-left), and unidentifiable (lower-right). Error bars denote $95 \%$ confidence intervals for MAD across 50 replicates. Algorithm~\ref{alg:oaat} is used to acquire 200 parameters.}
    \label{synth_figs_compare}
\end{figure}

\begin{figure}[t]
    \centering
    \begin{subfigure}{0.32\textwidth}
        \includegraphics[width=1\textwidth]{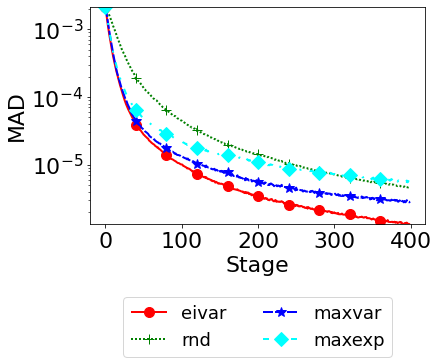}
    \end{subfigure}
    \begin{subfigure}{0.32\textwidth}
        \includegraphics[width=1\textwidth]{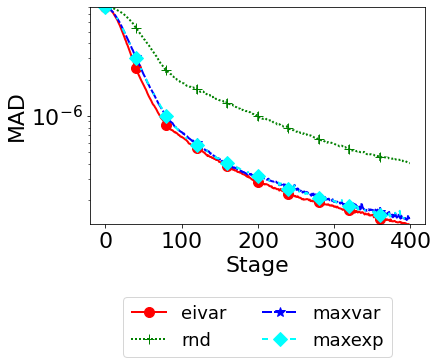}
    \end{subfigure}
    \begin{subfigure}{0.32\textwidth}
        \includegraphics[width=1\textwidth]{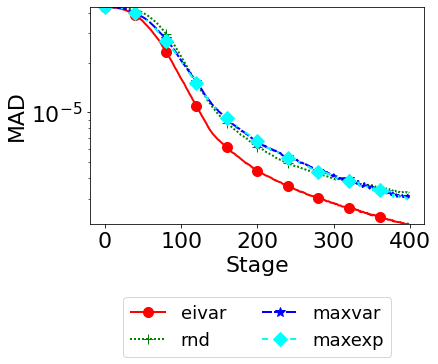}
    \end{subfigure}
    \caption{Comparison of four acquisition functions using three- (left), six- (center), and ten-dimensional (right) test functions. Error bars denote $95 \%$ confidence intervals for MAD across 50 replicates. Algorithm~\ref{alg:oaat} is used to acquire 400 parameters.}
    \label{synth_figs_compare2}
\end{figure}

Figures~\ref{synth_figs_compare}-\ref{synth_figs_compare2} summarize the accuracy results and indicate that EIVAR has the lowest MAD between the estimated posterior and the ground truth posterior for all seven test functions. 
From our observations of acquired parameters in low dimensional functions (see Figure~\ref{synth_figs_acq} in Appendix~\ref{sec:app2}), MAXVAR tends to ignore the regions where the cumulative uncertainty is larger and instead picks points at the boundary where the variance is larger. On the other hand, EIVAR recognizes that the acquired parameters impact the posterior estimation globally, and it acquires points both at the interior and boundary of the high-posterior region.
The RND strategy performs better than MAXEXP for the functions in Figure~\ref{synth_figs_compare}, especially during later stages of the sequential procedure. However, for higher-dimensional spaces (see left and center panels) it takes more time to accurately estimate the posterior for RND strategy as compared to other methods. The prior range of parameters is kept \tcb{the} same in three- and six-dimensional examples; \tcb{Figure~\ref{synth_figs_compare}} clearly shows that when the high posterior region becomes \tcb{smaller relative} to the support of the prior distribution, it becomes more challenging for RND to target the high posterior region. One disadvantage of EIVAR is the computational time needed to acquire a single parameter on higher-dimensional spaces since EIVAR requires an integral evaluation. \tcb{This limitation can be mitigated as $p$ grows, and in order to overcome the curse of dimensionality, one can consider sparse grids or other clever sampling or quadrature schemes for estimating higher dimensional integrals.}
Naturally, this computational expense is less of an issue when it is expensive to obtain a simulation output, as is the case described in the next subsection.

\subsection{Simulations in the Case of Discrepancy}
\begin{figure}[ht]
    \centering
    \begin{subfigure}{0.34\textwidth}
        \includegraphics[width=1\textwidth]{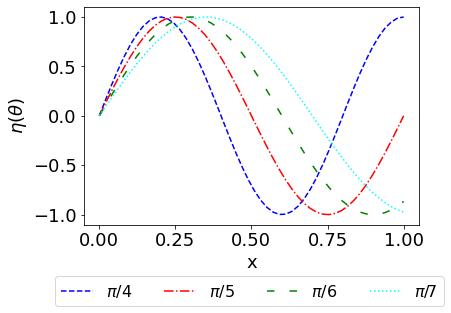}
    \end{subfigure}
     \begin{subfigure}{0.28\textwidth}
        \includegraphics[width=1\textwidth]{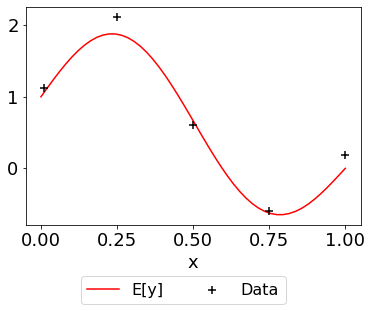}
    \end{subfigure}
    \begin{subfigure}{0.32\textwidth}
        \includegraphics[width=1\textwidth]{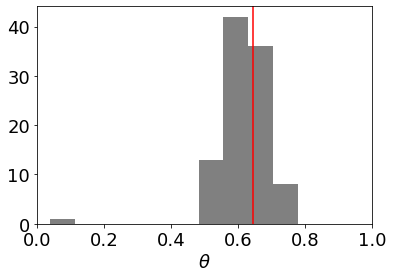}
    \end{subfigure}
    \caption{\tcb{An illustrative example in the case of discrepancy term with $\y = \model(\theta=\frac{\pi}{5}) + \mathbf{b} + \epsilonv$, $\model(\theta) = (\eta(x_1, \theta), \ldots, \eta(x_5, \theta))$, and $\mathbf{b} = (b(x_1), \ldots, b(x_5))$ where $\eta(x, \theta) = \sin(10 x \theta)$ and $b(x) = 1 - \frac{1}{3} x - \frac{2}{3} x^2$ at design points $(x_1, x_2, x_3, x_4, x_5) = (0, 0.25, 0.50, 0.75, 1)$. In the left panel, lines correspond to the simulation outputs at $\theta \in \{\frac{\pi}{4}, \frac{\pi}{5}, \frac{\pi}{6}, \frac{\pi}{7}\}$. In the middle panel, the red line corresponds to $\mathbb{E}(\y) = \model(\theta=\frac{\pi}{5}) + \mathbf{b}$, and the black markers show the observed data with $\sigma_\varepsilon^2 = 0.2^2$. The right panel illustrates the distribution of 100 acquired parameters along with the red vertical line at the least-squares fit parameter value minimizing the sum of the squared errors between $\y$ and $\mathbb{E}(\y)$.}}
    \label{fig:biastoy}
\end{figure}
\tcb{The results in Section~\ref{sec:syntheticresults} assume known $\Sigmav$ with no discrepancy between the simulation model and the observed data. In this section, we observe the acquired parameters via the EIVAR criterion when the covariance matrix of errors between the observed data and the simulation model is unknown. As an illustration, consider the example in Figure~\ref{fig:biastoy} with a single parameter in $[0, 1]$; see also the illustrative example presented in \cite{Koermer2023} for details. The left panel of Figure~\ref{fig:biastoy} shows simulation outputs at four different parameter values. In the middle panel, the mean of the observed data $\mathbb{E}(\y) = \model(\theta) + \mathbf{b}$ and the observed data at five design points on an equally spaced grid are shown. We set the initial sample size $n_0 = 10$, and acquire 100 parameters using the procedure described in Section~\ref{sec:KOH} to account for the observation noise and the uncertainty in the model in the form of discrepancy. We adopt a covariance form $\Sigmav^e$ defined by $\Sigmav^e_{i,j} = \sigma_\varepsilon^2 \delta_{i,j} + \sigma_b^2 \exp(-\lambda|x_i - x_j|)$ for $i, j = 1, \ldots, d$, where $\delta_{i,j}$ is a Kronecker delta function, and the likelihood in Equation~\eqref{eq:bias_like} is optimized to find the plausible estimates of ancillary parameters $\thetav^e = [\sigma_\varepsilon^2, \sigma_b^2, \lambda]$ at each stage. The least-squares fit parameter value is obtained by minimizing the sum of the squared errors between $\y$ and $\mathbb{E}(\y)$. In this example, the EIVAR criterion with unknown covariance matrix is able to target the high posterior region by acquiring almost all parameters concentrated around the least-squares fit parameter value.} 

%using SciPy/Python via \texttt{scipy.optimize.L-BFGS-B}
\subsection{Application to a Reaction Model}
\label{sec:physicsexample}

We now illustrate our sequential strategy with a nuclear physics model to predict differential cross section as a function of angle. Different parameterizations of the optical potential can result in different angular cross sections, and
the reaction code \texttt{FRESCOX} \citep{fresco} takes the optical model parameters as input and generates the corresponding cross sections across angles from $0^{\circ}$ to $180^{\circ}$ \citep{Lovell2017}. 
This case study uses the elastic scattering data for the 48Ca(n,n)48Ca reaction presented in \cite{Lovellphdthesis} to obtain the best-fit parameterization for optical parameters (see Appendix~B in \cite{Lovellphdthesis} for details). 
To simplify the example, some parameters from \cite{Lovellphdthesis} are fixed, since these parameters reproduce the data $\y$ well, while the remaining three parameters are acquired with the proposed approach. Table~\ref{tab:fresco_parameter} summarizes the parameters used in this case study and their prior ranges.
\begin{table}[t]
    \begin{center}
        \begin{tabular}{c c c c}
        Label & Input & Range & Unit  \\ \hline
        $\theta_1$ & $V$: depth (real volume) & [40, 60] & MeV \\
        $\theta_2$ & $r$: radii & [0.7, 1.2] & fm \\
        $\theta_3$ & $W_s$: depth (imaginary surface) & [2.5, 4.5] & MeV \\
        \end{tabular}
    \end{center}
\caption{Parameters and their ranges for the 48Ca(n,n)48Ca reaction model.}
\label{tab:fresco_parameter}
\end{table}

The elastic scattering data $\y$ is available only at $15$ different degree values, and the simulation outputs (i.e., cross sections) are obtained at the corresponding degree values. \tcb{We assume a known diagonal covariance matrix such that $\Sigmav_{i,i} = 0.1$ for $i = 1, \ldots, 15$}. We take $32$ initial samples and obtain $400$ parameters with each acquisition function. Each of the acquired parameters is chosen from the candidate list $\mathcal{L}_t$ of size 100, which is generated by sampling from the uniform prior at each stage $t$. First, we compare different acquisition strategies with Algorithm~\ref{alg:oaat} (e.g., $b=1$). Then, we focus on batch strategies with EIVAR and obtain the results with $b \in \{2, 4, 8, 16, 32\}$ via a batch synchronous procedure presented in Algorithm~\ref{alg:batch} (the number of workers equals to $b$ to run simulations in parallel). Figure~\ref{fresco_result} summarizes the results with both sequential and batched approaches. 

\begin{figure}[t]
    \centering
    \begin{subfigure}{1\textwidth}    \centering
        \includegraphics[width=0.4\textwidth]{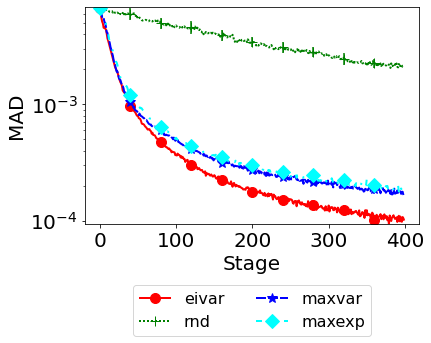}
    \end{subfigure}
    \begin{subfigure}{0.4\textwidth}
        \includegraphics[width=1\textwidth]{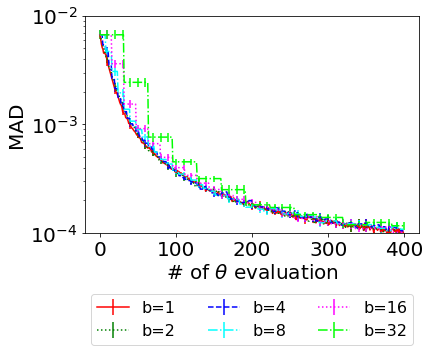}
    \end{subfigure}
        \begin{subfigure}{0.4\textwidth}
        \includegraphics[width=1\textwidth]{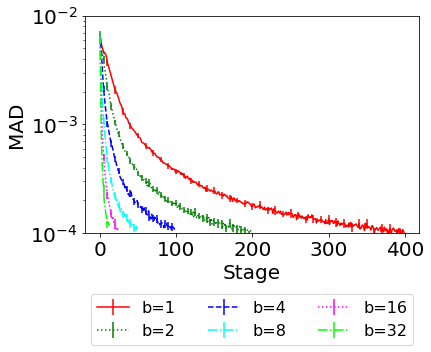}
    \end{subfigure}
    \caption{Results of physics reaction model. Error bars denote $95 \%$ confidence intervals for MAD across 50 replicates. Panel in the first row compares different acquisition functions; panels in the second row compare different batch sizes for the proposed EIVAR criterion with two perspectives.}
    \label{fresco_result}
\end{figure}
As illustrated in the first row of Figure~\ref{fresco_result}, with the sequential procedure the posterior is well-estimated with EIVAR as compared to the other methods, whereas the prediction accuracy with RND does not even improve significantly over time. Figure~\ref{cross_sections} illustrates the cross sections of the acquired parameters for RND, MAXVAR, and EIVAR. Since RND does not consider the system data (e.g., black plus markers in Figure~\ref{cross_sections}), most of the cross sections are placed onto the outside of the region of interest. While MAXVAR considers densely the regions around the data, EIVAR focuses on a wider region where the posterior is smaller but not exactly zero. Therefore, EIVAR covers the space well enough by not considering the zero-posterior region but at the same time distributing the acquired parameters around both high and low posterior regions to better learn the posterior. Finally, we observe the distribution of acquired parameters with EIVAR in Figure~\ref{eivar_hist} (left panel) along with the values of the best-fit parameters. Since the best-fit parameters are included into region specified with the acquired parameters, EIVAR is also able to identify the highest posterior region.
\begin{figure}[t]
    \centering
        \begin{subfigure}{0.32\textwidth}
        \includegraphics[width=1\textwidth]{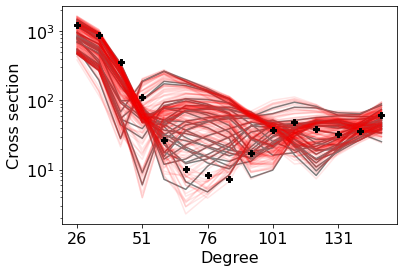}
    \end{subfigure}
    \begin{subfigure}{0.32\textwidth}
        \includegraphics[width=1\textwidth]{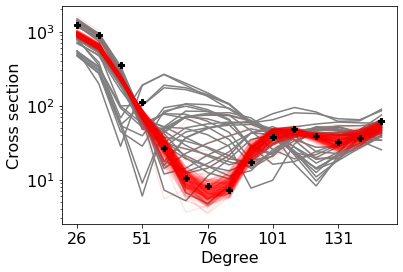}
    \end{subfigure}
    \begin{subfigure}{0.32\textwidth}
        \includegraphics[width=1\textwidth]{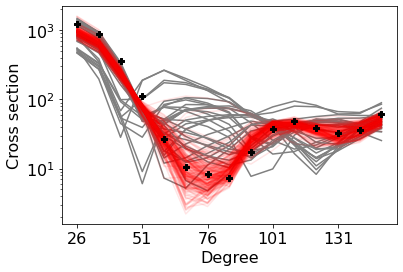}
    \end{subfigure}
    \caption{Cross sections (light-red lines) obtained with the sequential strategy via acquisition functions RND (left), MAXVAR (center), and EIVAR (right). Light-gray lines illustrate the cross sections for the initial sample. Black plus markers represent the observed data $\y$.}
    \label{cross_sections}
\end{figure}
\begin{figure}[t]
    \centering
    \begin{subfigure}{0.46\textwidth}
        \includegraphics[width=1\textwidth]{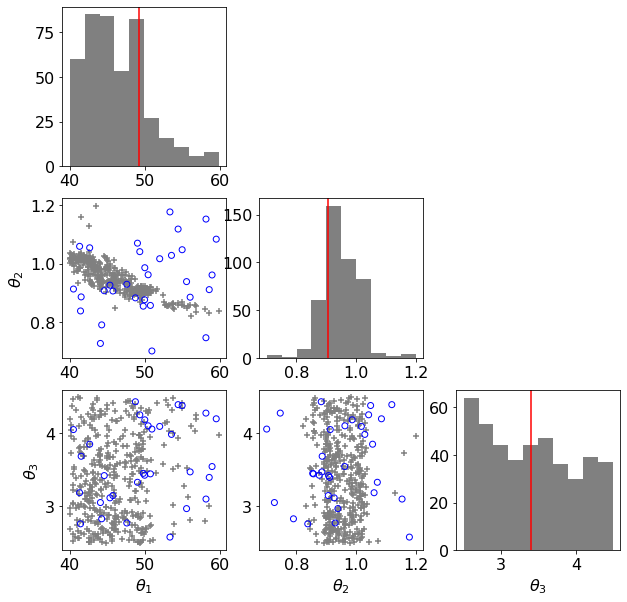}
    \end{subfigure}
    \begin{subfigure}{0.46\textwidth}
        \includegraphics[width=1\textwidth]{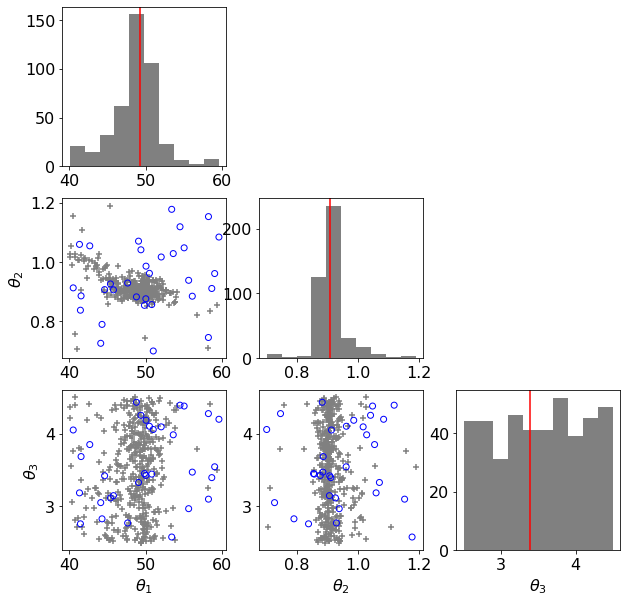}
    \end{subfigure}
    \caption{\tcb{Pairwise plot for acquired parameters with EIVAR (gray plus markers) with known (left panel) and unknown (right panel) covariance matrix. Blue circles are the initial sample obtained with LHS. The red line corresponds to the parameter values obtained with the best-fit parameterization.}}
    \label{eivar_hist}
\end{figure}

The results of the batch synchronous approach with batch sizes $b \in \{2, 4, 8, 16, 32\}$ are provided in the second row of Figure~\ref{fresco_result}. 
While the left panel summarizes MAD after obtaining each simulation output ($x$-axis represents each of 400 simulation evaluations), the right panel shows MAD after each stage ($x$-axis represents each stage). 
One can consider the figure in the right panel as a proxy to the total wall-clock time if the simulation was the dominant computational cost. 
Since 400 parameters are acquired with each of the batch sizes, sequential approaches with larger batch sizes terminate in fewer batches. 
Although the final accuracy of different batch sizes does not differ significantly between each other, there is a significant difference between the sequential strategy with $b=1$ and the batched strategy with $b=32$ up until 200 simulation evaluations are completed. 
This indicates that the quality of emulator with the kriging-believer strategy starts degrading with an increasing batch size, especially during the early periods of data collection when the acquired parameters and their corresponding simulation outputs are not enough to explore the space. Thus, a user should decide on the best batch size considering the trade-off between the computational time and accuracy.

\tcb{All the results so far use the known diagonal covariance matrix $\Sigmav$ with diagonal elements of 0.1 to acquire parameters. 
Although the assumption comes from the experts' best understanding of the experiment, one can also use the procedure presented in Section~\ref{sec:KOH} to learn aspects of the covariance. 
To illustrate this, we use the reaction code \texttt{FRESCOX} with the same settings introduced above except that we assume an unknown covariance $\Sigmav^e$ such that $\Sigmav^e_{i,j} = \sigma_\varepsilon^2 \delta_{i,j} + \sigma_b^2 \exp(-\lambda|x_i - x_j|)$, for $i, j = 1, \ldots, d$, to account for both the observation noise and the model uncertainty in the form of discrepancy where $\thetav^e = (\sigma_\varepsilon^2, \sigma_b^2, \lambda)$. The observed data $\y$ is collected at 15 design points (i.e., at angles) $(x_1, \ldots, x_{15}) = (26, 31, 41, 51, 61, 71, 76, 81, 91, 101,$ $111, 121, 131, 141, 151)$. The right panel of Figure~\ref{eivar_hist} shows the distribution of acquired parameters with an unknown covariance matrix and demonstrates that EIVAR is able identify the high posterior region in this situation as well. Figure~\ref{cross_sections_bias} illustrates the cross sections of the acquired parameters; similar to the known covariance matrix situation, EIVAR focuses on a wider region as compared to the other approaches tested. Additionally, we have tested EIVAR with an unknown covariance matrix $\Sigmav^e$ to incorporate only the observation noise such that $\Sigmav^e_{i,j} = \sigma_\varepsilon^2 \delta_{i,j}$, for $i, j = 1, \ldots, d$. For this case, the cross sections and the pairwise plot of acquired parameters are similar to the ones obtained with the known covariance illustrated in Figure~\ref{cross_sections} (right panel) and Figure~\ref{eivar_hist} (left panel), respectively. Although the resulting cross sections cover a similar region in all three of these cases, the parameters (especially $\theta_1$ and $\theta_2$) are tightly acquired around the best-fit parameters by the unknown covariance with the discrepancy term assumption (as in the right panel of Figure~\ref{eivar_hist}) indicating that the discrepancy term helps better constrain the range of acquired parameters.}
\begin{figure}[t]
    \centering
        \begin{subfigure}{0.45\textwidth}
        \includegraphics[width=1\textwidth]{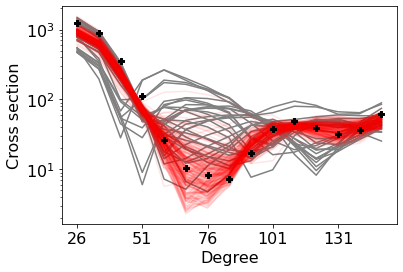}
    \end{subfigure}
    \caption{\tcb{Cross sections (light-red lines) obtained with the sequential strategy via EIVAR with the discrepancy term. Light-gray lines illustrate the cross sections for the initial sample. Black plus markers represent the observed data $\y$. }}
    \label{cross_sections_bias}
\end{figure}

\section{Conclusion}

This article proposes a novel sequential Bayesian experimental design procedure to estimate the posterior that is known up to a proportionality constant. Our approach selects the most informative parameters using previous simulation evaluations and \tcb{places a majority of them in the} high-posterior region, which facilitates the learning of the posterior region. The proposed approach can be advantageous especially when dealing with computationally expensive simulation models since it requires fewer number of simulation evaluations than space-filling approaches to learn the posterior at the same accuracy level. We extend a single acquisition case to batch acquisitions as well. Both sequential and batched settings require only the simulation model and the prior distribution for the parameters as input.
Thus, practitioners can easily use the proposed approaches for a wide variety of calibration applications. 

In this paper, a novel EIVAR criterion is derived for calibrating computationally expensive simulation models, and its closed-form expression makes it appealing for a variety of future work. Inspiring from the aggregate variance criterion when building an emulator with globally high prediction accuracy, one area is to extend the criterion for stochastic simulation models \citep{Binois2019}. In parallel to this, our proposed EIVAR criterion can be further developed to decide whether a new parameter is acquired for replication or exploration of the posterior surface when calibrating stochastic simulation models. Another area is to replace the aforementioned emulator with enhanced emulators such as deep GPs \citep{Sauer2022}, and the parameters can be selected via EIVAR sequentially for the purpose of calibration when dealing with complex response surfaces. \tcb{Moreover, in the presence of non-identifiability, multiple parameter sets may result in high posterior density, and this issue may be diminished or controlled with the right active learning criteria. In such a setting, one can utilize both the proposed EIVAR and the new active learning criteria in an hybrid way to better estimate the posterior.} \tcb{In the case where a simulation model returns one-dimensional output $\eta(\xv, \thetav)$ at a design point $\xv$ and calibration parameter $\thetav$, one can consider acquiring a design point and a calibration parameter simultaneously. For this new setup, one can consider replacing the proposed emulation strategy with a global GP emulator. Since our results do not depend on the way the emulator is built, our derivations can still be useful for this setting with minor adjustments to the resulting statements and proofs. Acquiring design points and the calibration parameters simultaneously would be especially of interest when a physical experimental design is the goal and the sequential collection of data $\y$ is possible. Additionally, in this paper, our results assume that the residual error is sampled from a multivariate normal distribution; another line of future development would be to derive the EIVAR criterion for different distributions of the residual error. }

Batch strategies allow computationally expensive simulation models to be run in parallel.
In cases where the run times of each simulation evaluation are roughly equal, the batch synchronous sequential approach described here is effective to reduce the wall-clock time. 
However, if the run times in a batch vary, the batch synchronous approach may lead to inefficient utilization of the computational resources since workers are left idle while waiting for the slowest job in the batch to finish. In order to improve the utilization of parallel computing resources, we can run function evaluations asynchronously: as soon as $b$ workers complete their jobs, $b$ new parameters are chosen and sent to workers choose new tasks for them (with $b < \text{the number of workers}$). A batch asynchronous procedure has received little attention as compared to synchronous one. A more-detailed performance analysis for both the synchronous and asynchronous procedures and a guideline for the selection of the best batch size would be of \tcb{great interest to practitioners} as a future study.

\section*{Acknowledgement}
%The authors are grateful to Xxx Xxxxxx and Xxxxxxxx X. Xxxxx for taking time to provide us physics reaction model.
The authors are grateful to Amy Lovell and Filomena M. Nunes for taking time to provide us physics reaction model.

\appendix

\section{Appendix}
\label{sec:app}

\subsection{Comparison of EIVAR with Common Acquisition Functions}
\label{sec:ei_imse_eivar}

\tcb{The expected improvement (EI) and the integrated mean squared error (IMSE) are the two common acquisition functions in the literature. While the EI criterion acquires the set of parameters offering the greatest expected improvement over the current best parameter to maximize a function (e.g., the posterior), IMSE selects the set of parameters that minimizes the overall uncertainty in the estimate of the simulation model at each stage to build a globally accurate emulator. 
In this section, we demonstrate the differences between EIVAR, IMSE, and EI acquisition functions with a one-dimensional illustrative example. }

\tcb{
To implement the EI criterion, a GP emulator of the posterior is obtained at each stage $t$ (see Section \ref{sec:roleGP} for an example), and then EI uses
\begin{equation}\label{eq:ei}
    \thetav^{\rm new} \in \argmax_{\thetav \in \mathcal{L}_t} \left((\breve{m}_t(\thetav) - p_{\max}) \Phi\left(\frac{\breve{m}_t(\thetav) - p_{\max}}{\breve{\sigma}_t(\thetav)}\right) + \breve{\sigma}_t(\thetav) f_{\mathcal{N}}\left(\frac{\breve{m}_t(\thetav) - p_{\max}}{\breve{\sigma}_t(\thetav)}; \, 0, \, 1\right)\right),
\end{equation}
where $p_{\max} = \max_{i=1,\ldots,n_t} \tilde{p}\left(\thetav_i|\y\right)$, $\Phi(\cdot)$ is the cumulative density function of the standard normal distribution, and $\breve{m}_t(\thetav)$ and $\breve{\sigma}^2_t(\thetav)$ are the prediction mean and variance resulting from a GP model for the posterior. To compute the IMSE at any parameter, an emulator of the simulation model is obtained at each stage as described in the paper and the parameter that minimizes the total uncertainty of the emulator is acquired via
\begin{equation}\label{eq:imse}
    \thetav^{\rm new} \in \argmin_{\thetav^* \in \mathcal{L}_t} \sum_{\thetav \in \Theta_{\rm ref}} \varsigma^2_{t+1}(\thetav),
\end{equation}
where $\varsigma^2_{t+1}(\thetav)$ depends on the chosen parameter $\thetav^*$ and is obtained as in Equation~\eqref{var_future} except that the index $j$ is dropped as we only explore the IMSE criterion for $d=1$.
} 
%$\mathbb{E}\left[\max\left(\tilde{p}\left(\thetav|\y\right) - p_{\max}, 0\right)\right]$ where 

\tcb{Figure~\ref{fig:MAD_EI_IMSE_EIVAR} summarizes the accuracy results across 50 replicates to acquire 10 parameters and Figure~\ref{fig:IMSEvsEIVAR} illustrates 10 acquired points for a single replication using the one-dimensional example with EI (first column), IMSE (second column), and EIVAR (third column) acquisition functions to gain insights about how the acquisition functions perform. The EI criterion acquires hyper-localized parameters and results in a myopic focus on the highest posterior region. Moreover, EI is not able to explore the medium-level posterior region, and thus the posterior is not correctly estimated at the corresponding parameter values. On the other hand, IMSE acquires parameters to build an emulator with a high prediction accuracy on the entire support of the parameters. Consequently, IMSE acquires many parameters at near-zero posterior regions, and it does not bring any additional value for estimating the posterior. This can be problematic especially if the available computational budget is limited and the simulation model is expensive to evaluate since the computational resources would be wasted for evaluations outside of the region of interest. In contrast, EIVAR faithfully acquires parameters from the nonzero posterior region to better estimate the posterior. Moreover, since IMSE does not specifically target the high posterior region, the uncertainty is higher and the predictive accuracy is lower on the high posterior region as compared to EIVAR. }
\begin{figure}[h!]
    \centering
    \begin{subfigure}{0.4\textwidth}
        \includegraphics[width=1\textwidth]{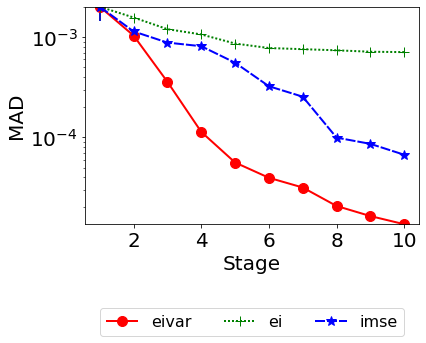}
    \end{subfigure}
        \caption{\tcb{Comparison of three acquisition functions using a test function $\eta(\theta) = \sin(\theta) + 0.5 \theta$, observation $y=-5$, and observation variance one. Algorithm~1 is used to acquire 10 parameters.}}
    \label{fig:MAD_EI_IMSE_EIVAR}
\end{figure}

\begin{figure}[h!]
    \centering
    \begin{subfigure}{0.32\textwidth}
        \includegraphics[width=1\textwidth]{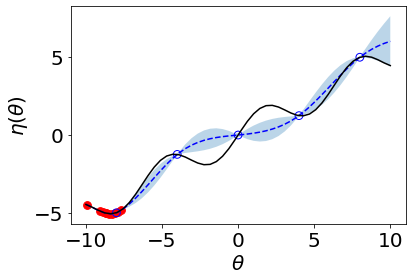}
    \end{subfigure}
    \begin{subfigure}{0.32\textwidth}
        \includegraphics[width=1\textwidth]{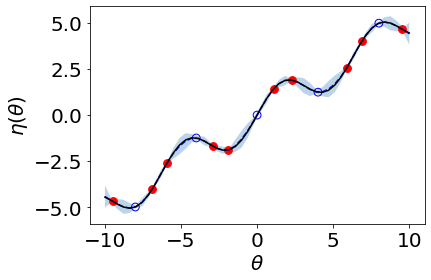}
    \end{subfigure}
    \begin{subfigure}{0.32\textwidth}
        \includegraphics[width=1\textwidth]{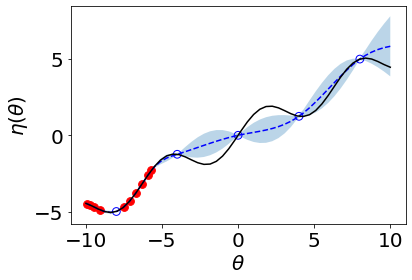}
    \end{subfigure}
    \begin{subfigure}{0.32\textwidth}
        \includegraphics[width=1\textwidth]{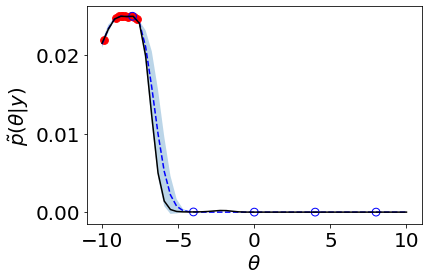}
    \end{subfigure}
    \begin{subfigure}{0.32\textwidth}
        \includegraphics[width=1\textwidth]{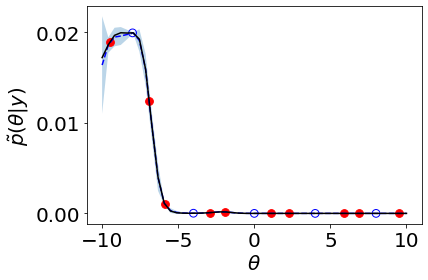}
    \end{subfigure}
    \begin{subfigure}{0.32\textwidth}
        \includegraphics[width=1\textwidth]{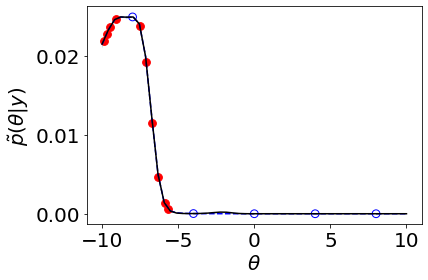}
    \end{subfigure}
    \caption{\tcb{10 samples (red dots) with EI (first column), IMSE (second column), and EIVAR (third column) acquisition functions for a test function $\eta(\theta) = \sin(\theta) + 0.5 \theta$, observation $y=-5$, and observation variance one for $\theta \in [-10, 10]$. Red dots indicate the acquired parameters and the black line illustrates the simulation model (first row) and the posterior (second row). The blue dashed-line shows the prediction mean and the shaded area illustrates two predictive standard deviation from the mean. Blue circles illustrate the 5 samples used to initiate the proposed procedure.}}
    \label{fig:IMSEvsEIVAR}
\end{figure}

\subsection{Additional Details and Experiments}
\label{sec:app2}
\begin{figure}[h!]
    \centering
    \begin{subfigure}{0.32\textwidth}
        \includegraphics[width=1\textwidth]{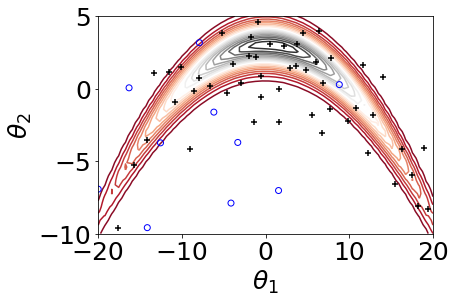}
    \end{subfigure}
    \begin{subfigure}{0.32\textwidth}
        \includegraphics[width=1\textwidth]{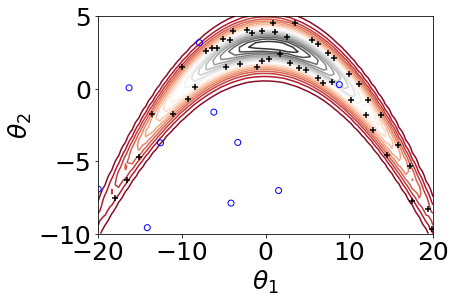}
    \end{subfigure}
    \begin{subfigure}{0.32\textwidth}
        \includegraphics[width=1\textwidth]{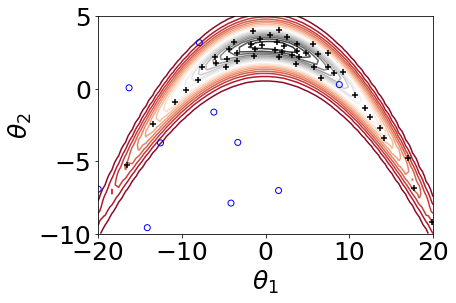}
    \end{subfigure}
    \begin{subfigure}{0.32\textwidth}
        \includegraphics[width=1\textwidth]{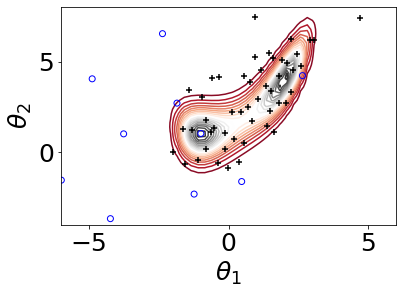}
     \end{subfigure}
    \begin{subfigure}{0.32\textwidth}
        \includegraphics[width=1\textwidth]{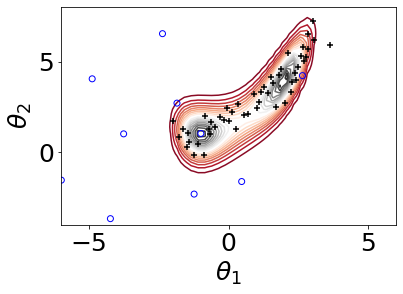}
     \end{subfigure}
    \begin{subfigure}{0.32\textwidth}
        \includegraphics[width=1\textwidth]{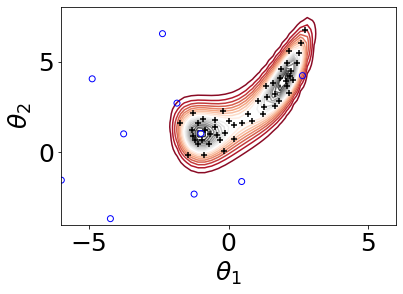}
     \end{subfigure}
    \begin{subfigure}{0.32\textwidth}
        \includegraphics[width=1\textwidth]{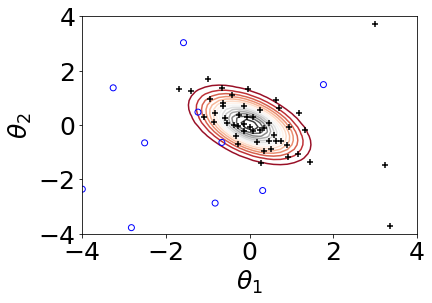}
     \end{subfigure}
    \begin{subfigure}{0.32\textwidth}
        \includegraphics[width=1\textwidth]{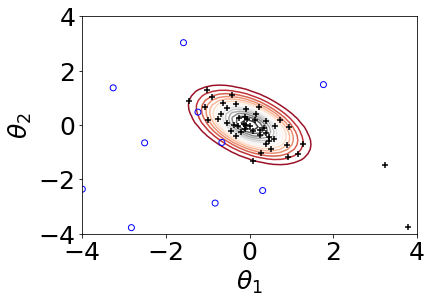}
     \end{subfigure}
    \begin{subfigure}{0.32\textwidth}
        \includegraphics[width=1\textwidth]{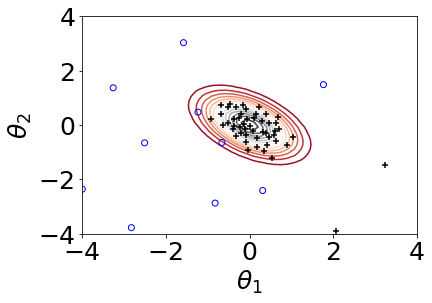}
     \end{subfigure}
         \begin{subfigure}{0.32\textwidth}
        \includegraphics[width=1\textwidth]{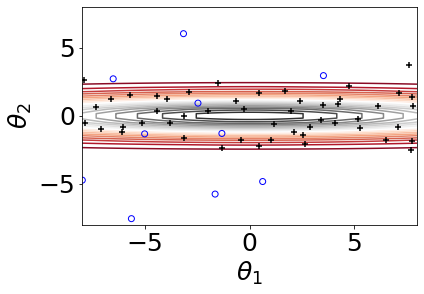}
     \end{subfigure}
    \begin{subfigure}{0.32\textwidth}
        \includegraphics[width=1\textwidth]{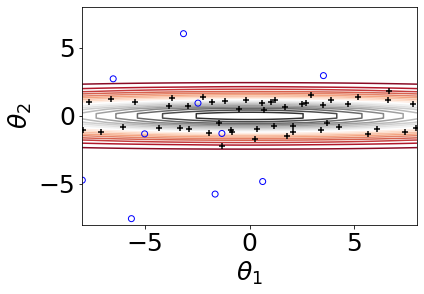}
     \end{subfigure}
    \begin{subfigure}{0.32\textwidth}
        \includegraphics[width=1\textwidth]{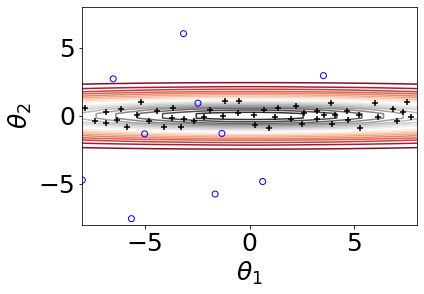}
     \end{subfigure}
    \caption{50 samples (plus markers) with different acquisition functions: EIVAR (first column), MAXVAR (second column), MAXEXP (last column). Test functions illustrated in Figure~\ref{synth_figs}: banana (first row), bimodal (second row), unimodal (third row), and unidentifiable (last row). Blue circles illustrate 10 samples used to initiate the proposed procedure.}
    \label{synth_figs_acq}
\end{figure}
In this section, we provide the details of the test functions used in our numerical experiments.
\begin{itemize}
    \item For the banana function (first plot in Figure~\ref{synth_figs}), we consider the unnormalized posterior given by
    \[
        \tilde{p}(\thetav|\y) = (2\pi)^{-1}|\Sigmav|^{-1/2} \exp\bigg\{-\frac{1}{2} \frac{(\theta_1-y_1)^2}{100} -\frac{1}{2} (\theta_2 + 0.03 \theta_1^2 - y_2)^2 \bigg\} p(\thetav)
    \]
    where $\theta_1 \in [-20, 20]$ and $\theta_2 \in [-10, 5]$. In the calibration context, we take $\thetav = (\theta_1, \theta_2)$ and $\model(\thetav) = (\theta_1, \theta_2 + 0.03 \theta_1^2)$ with $\y = (y_1, y_2) = (0, 3)$ and $\Sigmav$ is a diagonal matrix with diagonal elements $(10^2, 1)$.
    \item For the bimodal function (second plot in Figure~\ref{synth_figs}), we consider the unnormalized posterior given by
    \[
        \tilde{p}(\thetav|\y) = (2\pi)^{-1}|\Sigmav|^{-1/2} \exp\bigg\{-\frac{\sqrt{0.2}}{2} (\theta_2 - \theta_1^2 - y_1)^2 -\frac{\sqrt{0.75}}{2}  (\theta_2 - \theta_1 - y_2)^2 \bigg\} p(\thetav)
    \]
    where $\theta_1 \in [-6, 6]$ and $\theta_2 \in [-4, 8]$. In the calibration context, we take $\thetav = (\theta_1, \theta_2)$ and $\model(\thetav) = (\theta_2 - \theta_1^2, \theta_2 - \theta_1)$ with $\y = (y_1, y_2) = (0, 2)$ and $\Sigmav$ is a diagonal matrix with diagonal elements $(1/\sqrt{0.2}, 1/\sqrt{0.75})$.
    \item For the unimodal function (third plot in Figure~\ref{synth_figs}), we consider the unnormalized posterior given by
    \[
        \tilde{p}(\thetav|y) = \frac{1}{\sqrt{2\pi \sigma^2}} \exp\bigg\{-\frac{1}{2} \bigg(\frac{\theta_1^2 + \theta_1 \theta_2 + \theta_2^2 - y}{2}\bigg)^2 \bigg\} p(\thetav)
    \]
    where $\theta_1 \in [-4, 4]$ and $\theta_2 \in [-4, 4]$. In the calibration context, we take $\thetav = (\theta_1, \theta_2)$ and $\model(\thetav) = \theta_1^2 + \theta_1 \theta_2 + \theta_2^2$ with $y = -6$ and $\sigma^2 = 4$.
    \item For the unidentifiable function (last plot in Figure~\ref{synth_figs}), we consider the unnormalized posterior given by
    \[
        \tilde{p}(\thetav|\y) = (2\pi)^{-1}|\Sigmav|^{-1/2} \exp\bigg\{-\frac{1}{2} \frac{(\theta_1 - y_1)^2}{100} -\frac{1}{2}  (\theta_2 - y_2)^2 \bigg\} p(\thetav)
    \]
    where $\theta_1 \in [-8, 8]$ and $\theta_2 \in [-8, 8]$. In the calibration context, we take $\thetav = (\theta_1, \theta_2)$ and $\model(\thetav) = (\theta_1, \theta_2)$ with $\y = (y_1, y_2) = (0, 0)$ and $\Sigmav$ is a diagonal matrix with diagonal elements $(10^2, 1)$.
    
    \item For 3d function, we consider the unnormalized posterior given by
    \[
        \tilde{p}(\thetav|y) = (2\pi)^{-3/2}|\Sigmav|^{-1/2} \exp\left\{-\frac{1}{2} (\model(\thetav) - \y)^\top\Sigmav^{-1}(\model(\thetav) - \y) \right\} p(\thetav)
    \]
    where $\thetav = (\theta_1, \theta_2, \theta_3) \in [-4, 4]^3$, $\model(\thetav) = (\theta_1^2 + 0.5\theta_1 \theta_2 + 0.5\theta_1 \theta_3, \theta_2^2 + 0.5\theta_2 \theta_1 + 0.5\theta_2 \theta_3, \theta_3^2 + 0.5\theta_3 \theta_1 + 0.5\theta_3 \theta_2)$ with $\y = \textbf{0}_3$ and $\Sigmav$ is a diagonal matrix with diagonal elements each of which is 0.5.
    \item For 6d function, we consider the unnormalized posterior given by
    \[
        \tilde{p}(\thetav|y) = (2\pi)^{-3}|\Sigmav|^{-1/2} \exp\left\{-\frac{1}{2} \left(\model(\thetav) - \y\right)^\top\Sigmav^{-1}(\model(\thetav) - \y) \right\} p(\thetav)
    \]
    where $\thetav = (\theta_1, \theta_2, \ldots, \theta_6) \in [-4, 4]^6$, $\model(\thetav) = (0.5 \theta_1^2 + 0.5\theta_1 \theta_1 + \cdots + 0.5\theta_1 \theta_6, 0.5 \theta_2^2 + 0.5\theta_2 \theta_1 + \cdots + 0.5\theta_2 \theta_6, \ldots, 0.5 \theta_6^2 + 0.5\theta_6 \theta_1 + \cdots + 0.5\theta_6 \theta_6)$ with $\y = \textbf{0}_6$ and $\Sigmav$ is a diagonal matrix with diagonal elements each of which is 0.5.
    \item For 10d function, we consider the unnormalized posterior given by
    \[
        \tilde{p}(\thetav|y) = (2\pi)^{-5}|\Sigmav|^{-1/2} \exp\left\{-\frac{1}{2} \left(\model(\thetav) - \y\right)^\top\Sigmav^{-1}(\model(\thetav) - \y) \right\} p(\thetav)
    \]
    where $\thetav = (\theta_1, \theta_2, \ldots, \theta_{10}) \in [-2, 2]^{10}$, $\model(\thetav) = (0.5 \theta_1^2 + 0.5\theta_1 \theta_1 + \cdots + 0.5\theta_1 \theta_6, 0.5 \theta_2^2 + 0.5\theta_2 \theta_1 + \cdots + 0.5\theta_2 \theta_6, \ldots, 0.5 \theta_{10}^2 + 0.5\theta_{10} \theta_1 + \cdots + 0.5\theta_{10} \theta_{10})$ with $\y = \textbf{0}_{10}$ and $\Sigmav$ is a diagonal matrix with diagonal elements each of which is 0.25.
\end{itemize}
Figure~\ref{synth_figs_acq} illustrates 50 acquired points for a single replication using two-dimensional synthetic functions to gain insights about how the acquisition functions perform.

\subsection{Additional Results with Minimum Energy Design}
\label{sec:app3}

We also include the minimum energy design (a.k.a.~MINED) criterion \citep{Joseph2015, Joseph2019} into the benchmark. To do that, we use the \texttt{R} package MINED \citep{minedR} to acquire 200 parameters, and replicate the results 50 times. We set the iteration number $K_\textrm{iter} = 8$ as illustrated in the two-dimensional example provided in the vignette. Since the acquired parameters are not ordered, we did not include MINED in the results presented in Figure~\ref{synth_figs_compare}, and instead we compare the results with a total of 200 parameters. Figure~\ref{synth_figs_compare_mined} illustrates the results for four toy problems presented in the paper. %Results indicate that EIVAR performs better than MINED, and MINED has larger variability than EIVAR. We also note that MINED requires evaluating the simulation model at least more than 1500 parameters to acquire 200 of them, whereas the proposed approach only requires $n + n_0$ evaluations (which is 210 in these experiments). Evaluating the simulation model many times can be problematic when run time of the simulation model is expensive.
\begin{figure}[h]
    \centering
    \begin{subfigure}{0.45\textwidth}
        \includegraphics[width=1\textwidth]{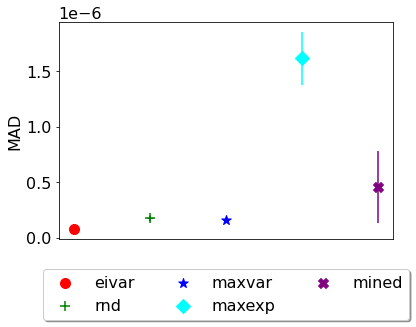}
    \end{subfigure}
    \begin{subfigure}{0.45\textwidth}
        \includegraphics[width=1\textwidth]{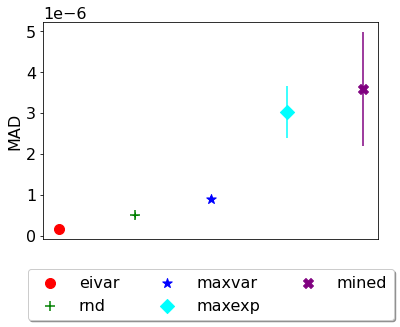}
    \end{subfigure}
    \begin{subfigure}{0.45\textwidth}
        \includegraphics[width=1\textwidth]{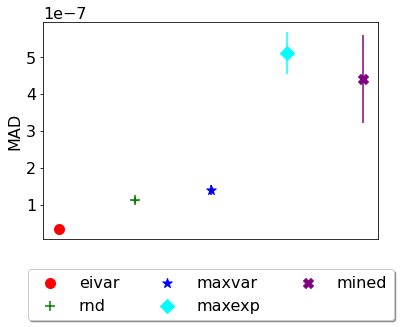}
    \end{subfigure}
    \begin{subfigure}{0.45\textwidth}
        \includegraphics[width=1\textwidth]{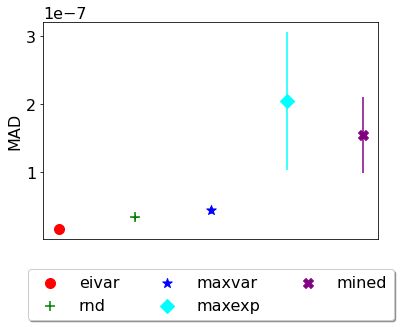}
     \end{subfigure}
    \caption{Comparison of different acquisition functions using the test functions illustrated in Figure~\ref{synth_figs}: banana (upper-left), bimodal (upper-right), unimodal (lower-left), and unidentifiable (lower-right).}
    \label{synth_figs_compare_mined}
\end{figure}

\subsection{Impact of Prior on EIVAR}
\label{sec:priorsensitivity}

\tcb{Similar to the other Bayesian calibration procedures, the prior probability density $p(\thetav)$ is critical in the proposed approach. Figure~\ref{fig:priorsensitivity} illustrates an example with two parameters $\thetav = (\theta_1, \theta_2)$ in $[-2, 5] \times [-2, 5]$ to test the impact of  prior on the proposed approach. In Figure~\ref{fig:priorsensitivity}, for all three cases, the likelihood is the same and the truncated Gaussian prior is used for each parameter with the same mean. We vary the variance of the prior to test the impact of prior on the acquired parameters. In all cases, EIVAR has places points to learn the posterior and majority of the acquired parameters are located onto the high posterior density regions. This example shows that our approach selects the most informative parameters consistently for different strength of prior density.}

\begin{figure}[h]
    \centering
    \begin{subfigure}{0.32\textwidth}
        \includegraphics[width=1\textwidth]{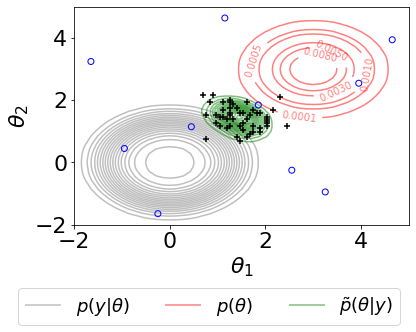}
    \end{subfigure}
    \begin{subfigure}{0.32\textwidth}
        \includegraphics[width=1\textwidth]{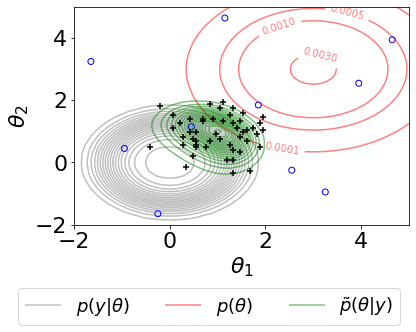}
    \end{subfigure}
    \begin{subfigure}{0.32\textwidth}
    \includegraphics[width=1\textwidth]{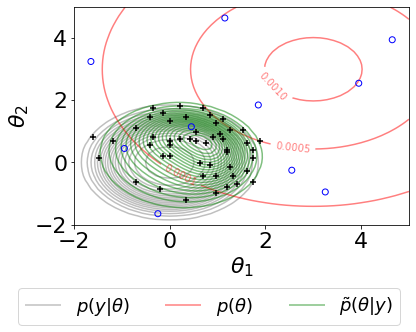}
    \end{subfigure}
    \caption{\tcb{50 samples (plus markers) with EIVAR acquisition function for a test function $\eta(\theta_1, \theta_2) = \theta_1^2 + \theta_2^2$, observation $y = 0$ and observation variance 2. Shown are the likelihood $p(y|\thetav)$ (gray contour plot), the prior $p(\thetav)$ (red contour plot), and the posterior $\tilde{p}(\thetav|y)$ (green contour plot) with varying values of variance for a truncated Gaussian prior: $0.5^2$ (left panel), $1$ (middle panel), and $2^2$ (right panel). Blue circles illustrate the 10 samples used to initiate the proposed procedure.}}
    \label{fig:priorsensitivity}
\end{figure}

%\label{sec:app4}

\clearpage
\bibliographystyle{JASA}
\bibliography{activelearning}

\end{document}